\documentclass[dvipsnames,11pt]{article}
\usepackage{amsmath,amssymb,amsthm,color,bm,mathrsfs,extarrows,tikz,graphicx,mathtools,enumitem,fancybox,makecell}
\usepackage[square,sort,comma,numbers]{natbib}
\usepackage{subcaption}
\usepackage[ruled]{algorithm2e}
\usepackage[margin=1in]{geometry}
\usepackage{xcolor,bbm}
\usepackage[pagebackref,hypertexnames=false]{hyperref}
\usepackage{comment}
\usepackage{soul}
\usepackage{thm-restate}
\usepackage{ifthen}
\usepackage{mathdots}
\usepackage{enumitem}
\usepackage{float}

\allowdisplaybreaks

\setlist[itemize]{itemsep=0pt}
\setlist[enumerate]{itemsep=0pt}
\hypersetup{colorlinks=true,urlcolor=blue,linkcolor=blue,citecolor=[rgb]{.42,.56,.14},}
\usetikzlibrary{decorations.pathreplacing,decorations.markings,decorations.pathmorphing,decorations.shapes,arrows.meta,positioning,patterns, quantikz2}

\usepackage[capitalise,nameinlink]{cleveref}
\Crefname{lemma}{Lemma}{Lemmas}
\Crefname{fact}{Fact}{Facts}
\Crefname{question}{Question}{Questions}
\Crefname{theorem}{Theorem}{Theorems}
\Crefname{corollary}{Corollary}{Corollaries}
\Crefname{claim}{Claim}{Claims}
\Crefname{example}{Example}{Examples}
\Crefname{problem}{Problem}{Problems}
\Crefname{definition}{Definition}{Definitions}
\Crefname{notation}{Notation}{Notations}
\Crefname{assumption}{Assumption}{Assumptions}
\Crefname{subsection}{Subsection}{Subsections}
\Crefname{section}{Section}{Sections}
\Crefformat{equation}{(#2#1#3)}

\newtheorem{theorem}{Theorem}[section]
\newtheorem*{theorem*}{Theorem}

\newtheorem*{proposition*}{Proposition}

\newtheorem*{property*}{Property}
\newtheorem{lemma}[theorem]{Lemma}
\newtheorem*{lemma*}{Lemma}
\newtheorem{corollary}[theorem]{Corollary}
\newtheorem*{corollary*}{Corollary}
\newtheorem*{conjecture*}{Conjecture}
\newtheorem{fact}[theorem]{Fact}
\newtheorem*{fact*}{Fact}

\newtheorem*{exercise*}{Exercise}

\newtheorem*{hypothesis*}{Hypothesis}

\theoremstyle{definition}
\newtheorem{definition}[theorem]{Definition}

\newtheorem{example}[theorem]{Example}

\newtheorem{exercise-easy}[theorem]{Exercise}
\newtheorem{exercise-med}[theorem]{Exercise}
\newtheorem{exercise-hard}[theorem]{Exercise$^\star$}
\newtheorem{claim}[theorem]{Claim}
\newtheorem*{claim*}{Claim}

\newtheorem{remark}[theorem]{Remark}
\newtheorem*{remark*}{Remark}

\newtheorem*{observation*}{Observation}

\DeclareSymbolFont{extraup}{U}{zavm}{m}{n}
\DeclareMathSymbol{\varheart}{\mathalpha}{extraup}{86}
\DeclareMathSymbol{\vardiamond}{\mathalpha}{extraup}{87}

\DeclareMathOperator*{\E}{\mathbb E}

\renewcommand{\Pr}{\operatorname*{\mathbf{Pr}}}

\newcommand{\eps}{\varepsilon}
\newcommand{\abs}[1]{\left| #1 \right|}
\newcommand{\vabs}[1]{\left\| #1 \right\|}

\newcommand{\pbra}[1]{\left( #1 \right)}
\newcommand{\sbra}[1]{\left[ #1 \right]}
\newcommand{\cbra}[1]{\left\{ #1 \right\}}
\newcommand{\floorbra}[1]{\left\lfloor #1 \right\rfloor}
\newcommand{\ceilbra}[1]{\left\lceil #1 \right\rceil}

\renewcommand{\mid}{\,\middle\vert\,}

\newcommand{\bin}{\{0,1\}}

\newcommand{\poly}{\mathsf{poly}}

\newcommand{\indicator}{\mathbbm{1}}

\newcommand{\supp}[1]{\mathsf{supp}\pbra{#1}}

\newcommand{\ac}{\mathsf{AC^0}}
\newcommand{\nc}{\mathsf{NC^0}}
\newcommand{\ext}{\mathsf{Ext}}
\newcommand{\rext}{\mathsf{rExt}}
\newcommand{\iso}{\mathsf{Iso}}
\newcommand{\add}{\mathsf{addr}}

\newcommand{\Fbb}{\mathbb{F}}

\newcommand{\Ecal}{\mathcal{E}}

\newcommand{\Hcal}{\mathcal{H}}

\newcommand{\Xcal}{\mathcal{X}}
\newcommand{\Ycal}{\mathcal{Y}}

\newcommand{\Ab}{\mathbf{A}}
\newcommand{\Bb}{\mathbf{B}}
\newcommand{\db}{\mathbf{d}}
\newcommand{\Db}{\mathbf{D}}
\newcommand{\Gb}{\mathbf{G}}
\newcommand{\gb}{\mathbf{g}}

\newcommand{\hb}{\mathbf{h}}
\newcommand{\Mb}{\mathbf{M}}
\newcommand{\Pb}{\mathbf{P}}

\newcommand{\Qb}{\mathbf{Q}}

\newcommand{\Sb}{\mathbf{S}}
\newcommand{\ub}{\mathbf{u}}
\newcommand{\Ub}{\mathbf{U}}
\newcommand{\xb}{\mathbf{x}}
\newcommand{\Xb}{\mathbf{X}}
\newcommand{\Yb}{\mathbf{Y}}
\newcommand{\Zb}{\mathbf{Z}}
\newcommand{\zb}{\mathbf{z}}
\newcommand{\Sigmab}{\mathbf{\Sigma}}

\newcommand{\tvdist}[1]{\vabs{#1}_\mathsf{TV}}

\renewcommand{\tilde}{\widetilde}

\title{Hard-to-Sample Distributions from Robust Extractors}

\author{
Farzan Byramji\thanks{UC San Diego. Email: \href{mailto:fbyramji@ucsd.edu}{\texttt{fbyramji@ucsd.edu}}. Supported by Simons Investigator Award \#929894, and NSF Awards CCF-2425349 and AF: Medium 2212136.}
\and 
Daniel M.~Kane\thanks{UC San Diego. Email: \href{mailto:dakane@ucsd.edu}{\texttt{dakane@ucsd.edu}}. Supported by NSF Medium Award CCF-2107547.}
\and 
Jackson Morris\thanks{UC San Diego. Email: \href{mailto:jrm035@ucsd.edu}{\texttt{jrm035@ucsd.edu}}.}
\and 
Anthony Ostuni\thanks{UC San Diego. Email: \href{mailto:aostuni@ucsd.edu}{\texttt{aostuni@ucsd.edu}}. Partially supported by Simons Investigator Award \#929894 and NSF Award CCF-2425349.}
}

\date{}

\begin{document}

\maketitle

\begin{abstract}
    We provide a unified method for constructing explicit distributions which are difficult for restricted models of computation to generate.
    Our constructions are based on a new notion of \emph{robust extractors}, which are extractors that remain sound even when a small number of points violate the min-entropy constraint.
    Using such objects, we show that for a broad range of sampling models (e.g., low-depth circuits, small-space sources, etc.), every output of the model has distance $1 - o(1)$ from our target distribution, qualitatively recovering essentially all previously known hardness results.
    Our work extends that of Viola (SICOMP '14), who developed an earlier unified framework based on traditional extractors to rule out sampling with very small error.
    
    As a further application of our technique, we leverage a recent extractor construction of Chattopadhyay, Goodman, and Gurumukhani (ITCS '24) to present the first explicit distribution with distance $1 - o(1)$ from the output of any low-degree $\mathbb{F}_2$-polynomial source.
    We note that a similar bound was obtained concurrently and independently by Khodabandeh and Shinkar (ECCC '26).
    We also describe a potential avenue toward proving a similar hardness result for $\mathsf{AC^0}[\oplus]$ circuits.
\end{abstract}

\section{Introduction}

The quest to prove unconditional hardness results for restricted models of computation (e.g., low-depth circuits) is one of the major programs of complexity theory.
This program saw a number of major successes in the 1980s when researchers exposed the limitations of $\ac$ and $\ac[p]$ circuits with simple and explicit hard-to-compute functions \cite{FSS84, Ajt83, yao1985separating, hastad1986almost, haastad1986computational, smolensky1987algebraic, razborov1987lower}.
Unfortunately, the current state of affairs on this front remains rather humble: even four decades later we still cannot rule out the preposterous claim that all of $\textsf{PSPACE}$ can be computed by polynomially-sized $\ac$ circuits with mod 6 gates.

It would be disingenuous, however, to suggest that no modern progress has occurred.
Even just the past few years have seen impressive developments in understanding the limitations of these models through the lens of pseudorandom generators \cite{DILV24, DH25, HL25, Vin25, LV25} (see also the survey \cite{HH24}), extractors \cite{chattopadhyay2024extractors,GGH+24, AGMR25}, and quantum analogs \cite{NPVY24, ADOY25, JTVW25, FGPT25, GMW26}.\footnote{This is only a minute fraction of contemporary work, and we encourage the interested reader to consult references within those cited.}

One particularly exciting line of inquiry is understanding the capabilities of these weak models to generate specific distributions (when taking random bits as input).
The landscape here is dramatically different from the traditional task of computing specific functions.
For example, if a circuit can compute a function $f\colon \bin^n \to \bin$, then it can output the uniform distribution over $(x,f(x))$ by simply computing $f$ on each random input.
The converse, however, is not true.
While \textsf{PARITY} has long been known to be difficult for $\ac$ circuits to compute \cite{FSS84, Ajt83, yao1985separating, hastad1986almost, haastad1986computational, smolensky1987algebraic, razborov1987lower}, one can produce the uniform distribution over input-output pairs by simply mapping the uniformly random bits $(x_1, x_2, \dots, x_{n+1})$ to $(x_1 \oplus x_2, x_2 \oplus x_3, \dots, x_{n}\oplus x_{n+1}, x_{n+1}\oplus x_1)$ \cite{babai1987random, boppana1987one}. 

Even still, for many common models including small $\ac$ circuits \cite{lovett2011bounded, beck2012large}, small-space sources \cite{chattopadhyay2022space}, and communication protocols \cite{AST+03, GW20, chattopadhyay2022space, yu2024sampling}, researchers have discovered explicit distributions which have total variation distance $1-o(1)$ from any distribution produced by the model.
In comparison to our knowledge of computation, the glaring shortcoming in our understanding of sampling is that we do not know of any explicit distribution which cannot be (approximately) produced by polynomially-sized $\ac$ circuits with mod $p$ gates for any prime $p$.

Beyond its inherent intrigue, the complexity of distributions has intimate connections with data structure lower bounds \cite{viola2012complexity, lovett2011bounded, beck2012large, viola2020sampling, chattopadhyay2022space, viola2023new, yu2024sampling, kane2024locality, alekseev2025sampling}, quantum supremacy \cite{WP26, kane2024locality, GKM+26}, explicit codes \cite{shaltiel2024explicit}, and learning theory \cite{KOW25}.
Moreover, the field has had a fruitful relationship with the study of \emph{extractors} \cite{trevisan2000extracting, viola2012extractors, de2012extractors, viola2014extractors, ball2025extractors, Sha25}, which are objects that convert particular sources of randomness into approximately uniform ones.

This work further develops the latter connection by introducing a general paradigm based on \emph{robust extractors} to prove sampling lower bounds.
Intuitively, one can view robust extractors as extractors which remain sound even when a small number of points violate the min-entropy condition.
There are several ways to formalize this notion, and we defer our exact definition and discussion of alternatives to \Cref{ssec:ext_and_rext}.
For now, we will say a robust extractor $\rext\colon \bin^n \to \bin^m$ for a class of distributions $\Xcal$ satisfies the following two conditions:
\begin{enumerate}
    \item \label{itm:ext_property} (Extractor Property) For every source $\Xb \in \Xcal$ with sufficient min-entropy, the output of $\rext$ on $\Xb$ is close in total variation (TV) distance to the uniform distribution over $\bin^m$, and

    \item (Robustness Property) For every source $\Xb \in \Xcal$, the probability that $\xb \sim \Xb$ lands in the set of $\Xb$'s low probability points and $\rext(\xb) = 0^m$ is not much more than $2^{-m}$.
\end{enumerate}
With these objects in hand, we can generically prove strong sampling lower bounds.
This result further advances a similar approach based on traditional extractors \cite{viola2014extractors}, which we discuss in greater detail in \Cref{sec:main_pf}.
Below, $\Ub^n$ denotes the uniform distribution over $\bin^n$, and for a function $f\colon\bin^n\to\bin^m$, $f(\Ub^n)$ denotes the output distribution of $f$ with $n$ uniformly random bits as input.

\begin{theorem}[Informal version of \Cref{thm:main}]\label{thm:main_informal}
    Let $\Xcal$ be a class of distributions over $\bin^{t(n+1)}$, and let $\Ycal$ be the class of distributions over $\bin^n$ obtained from ``low complexity'' functions of $\Xcal$.
    Suppose $\rext\colon \bin^n \to \bin^m$ is a sufficiently good robust extractor for $\Ycal$, and define the function $f\colon \bin^n \to \bin$ by $f(x) = \indicator(\rext(x)=0^m)$.
    If we consider the distribution
    \[
        \Db = (\Ub_1, \Ub_2, \dots, \Ub_t, f(\Ub_1), f(\Ub_2), \dots, f(\Ub_t)),
    \]
    where $\Ub_1, \dots, \Ub_t$ are independent copies of $\Ub^n$, then every source $\bf{X}\in \Xcal$ has TV distance
    \[
        \tvdist{\Xb - \Db} \ge 1 - o_{t,n}(1),
    \]
    where $o_{t,n}(1) \to 0$ as $t,n \to \infty$.
    
    Moreover, if $\tilde{\rext}\colon \bin^n \to \bin$ is a sufficiently good robust extractor (in a different parameter regime) for $\Ycal$, then $\tvdist{\Xb - (\Ub^n, \tilde{\rext}(\Ub^n))} \ge \frac{1}{4} - o_n(1)$.
\end{theorem}

\begin{remark}
    Robust extractors are the motivating object behind our approach, but actually a weaker one-sided notion in place of the typical extractor guarantee (\ref{itm:ext_property}) suffices for our purposes.
    In particular, \Cref{thm:main_informal} holds as long as $\rext$ on the uniform distribution assigns decent probability to $0^m$.
    The details can be found in \Cref{sec:main_pf}.
\end{remark}

We highlight that the ``moreover'' part of \Cref{thm:main_informal} can be viewed as a stronger type of lower bound than those for approximate computation.
If, for example, the functions $f, g\colon \bin^n \to \bin$ agree on a $(1-\delta)$-fraction of their inputs, then it is easy to generate the uniform distribution over $(\Ub^n, g(\Ub^n))$ to TV distance at most $\delta$ by simply computing $f$ on each random input.  

Using \Cref{thm:main_informal}, we can recover unconditional hardness results for sampling in essentially every model where they are currently known (albeit with weaker decay rates).
We are also able to provide the first explicit\footnote{One can essentially use a standard counting argument to show such a distribution exists; see \Cref{app:non-constructive}.} distribution which has distance $1-o(1)$ from the output of any low-degree $\Fbb_2$-polynomial source, making progress toward a similar result for $\ac[\oplus]$ circuits.\footnote{Recall these are $\ac$ circuits (discussed in \Cref{ssec:circuit}) with mod 2 gates.}
Previous arguments \cite{viola2014extractors, chattopadhyay2024extractors} only forbade distance $2^{-\Omega(n)}$.

\begin{theorem}[{Informal instantiations of \Cref{thm:main_informal}}]\label{thm:hard_dist_combined}
    Let $\Xcal$ be one of the following classes of sources over $\bin^N$:
    \begin{itemize}
        \item low-degree $\Fbb_2$-polynomial (see \Cref{thm:hard_dist_poly}), 
        \item local {\normalfont($\nc$)} (see \Cref{cor:hard_dist_nc0}), 
        \item circuit {\normalfont($\ac$)} (see \Cref{thm:hard_dist_ac0}), 
        \item communication (see \Cref{thm:hard_dist_communication}), 
        \item small-space (see \Cref{cor:hard_dist_small_space}), 
        \item Turing machine (see \Cref{cor:hard_dist_TM}).
    \end{itemize}
    Then every source $\bf{X}\in \Xcal$ satisfies
    \[
        \tvdist{\Xb - \Db} \ge 1 - o(1),
    \]
    where $\Db$ is defined by taking many independent copies of $(\Ub^n, \indicator(\rext(\Ub^n)=0^m))$ for an explicit robust extractor $\rext\colon\bin^n \to \bin^m$ (depending on the choice of $\Xcal$).
    
    Moreover, there exists an explicit robust extractor $\tilde{\rext}\colon \bin^{N-1} \to \bin$ (in a different parameter regime) such that every source $\bf{X}\in \Xcal$ satisfies
    \[
        \tvdist{\Xb - (\Ub^{N-1}, \tilde{\rext}(\Ub^{N-1}))} \ge \frac{1}{4} - o(1).
    \]
\end{theorem}

We emphasize that bounds of this form were already known (with better quantitative behavior) for local and circuit \cite{lovett2011bounded, beck2012large}, communication \cite{AST+03, GW20, chattopadhyay2022space, yu2024sampling}, small-space \cite{chattopadhyay2022space}, and Turing-machine sources \cite{viola2012extractors, chattopadhyay2022space}.
Our novel contribution is obtaining these bounds in a \emph{unified} way, as well as providing strong bounds for low-degree $\Fbb_2$-polynomial sources.

\paragraph*{Concurrent Work.}
In independent and concurrent work, Khodabandeh and Shinkar also investigate the complexity of sampling with low-degree $\Fbb_2$-polynomials \cite{MKS26}.
Using very different techniques from our own, they prove that the output of any such polynomial sampler has TV distance $1-o(1)$ from the $(1/3)$-biased product distribution.
As in the present work, the $o(1)$ term has suboptimal decay, but in the case of degrees 1, 2, and 3, they prove it can be taken to be $\exp(-\Omega(N))$, $\exp(-\Omega(\log(N) / \log\log(N)))$, and $\exp(-\Omega(\sqrt{\log\log(N)})$, respectively, which is stronger than we are able to obtain via our robust extractor framework.
(Note, however, that an extractor-based approach can match their bound for degree-1 polynomials -- see \Cref{sec:open_problems}.)

\paragraph*{Paper Organization.}
We state and prove \Cref{thm:main}, the precise version of \Cref{thm:main_informal}, in \Cref{sec:main_pf}, along with additional background and context.
We then review some preliminary material in \Cref{sec:prelim} before instantiating \Cref{thm:main} in \Cref{sec:hard_dist_instantiation} to obtain \Cref{thm:hard_dist_combined}.
More specifically, we address polynomial sources in \Cref{ssec:poly}; local and circuit sources in \Cref{ssec:circuit}; and communication, small-space, and Turing-machine sources in \Cref{ssec:communication}.
We conclude with some open problems in \Cref{sec:open_problems}.
Supplementary material on non-explicit constructions can be found in \Cref{app:non-constructive}.

\section{Background and Main Result}\label{sec:main_pf}

In this section, we state and prove our main result, \Cref{thm:main}.
Before doing so, it will be instructive to review existing arguments and their limitations to develop intuition for our approach.
For concreteness, we will focus our discussion on distributions generated by low-degree $\Fbb_2$-polynomials, but essentially all of the analysis holds more generally.
The meaning of any unfamiliar terminology or notation in this section can be found in \Cref{sec:prelim}.

Let $P \colon \Fbb_2^r \to \Fbb_2^n$ be a polynomial map defined by $n$ arbitrary degree-$d$ polynomials $\{p_i \colon \Fbb_2^r \to \Fbb_2\}_{i=1}^n$ acting on the same $r$ input bits, where we view $r$ as some arbitrarily large integer.
Our goal is to construct an explicit distribution $\Db$ over $\bin^n$ such that regardless of how the $p_i$'s are defined, the output distribution of $P$ on $r$ uniformly random bits, denoted $P(\Ub^r)$, has total variation distance at least $1 - o(1)$ from $\Db$.
(Here, we identify $\Fbb_2^n$ with $\bin^n$.)

\subsection{A Nonzero Lower Bound}\label{ssec:nonzero_LB}

We begin with the more modest goal of showing that $P$ cannot \emph{exactly} generate some explicit distribution $\Db$.
Here, we may invoke an argument of Viola \cite{viola2014extractors, viola2016quadratic} based on extractors, whose formal definition we now recall.

\begin{definition}[Extractor]
\label{def:extractor}
    A function $\ext\colon \bin^n \to \bin$ is an \emph{$(\eps,k)$-extractor} for a class $\Xcal$ of distributions over $\bin^n$ if for every source $\Xb \in \Xcal$ with min-entropy at least $k$, we have
    \[
        \Pr[\ext(\Xb) = 1] = \frac{1}{2} \pm \eps.
    \]
\end{definition}

We will take the hard distribution to be uniform over input-output pairs of an explicit $(\eps,k)$-extractor $\ext\colon \bin^{n-1} \to \bin$ for polynomial sources of degree $2d$, where $\eps$ is some small constant and $k = n - O(1)$.
That is, $\Db = (\Ub^{n-1}, \ext(\Ub^{n-1}))$.
Such extractors which are computable in time $\poly(n)$ are known \cite{chattopadhyay2024extractors}.
For clarity, we express $P(\Ub^r)$ similarly as $(Q(\Ub^r), q(\Ub^r))$, where $Q\colon \bin^r \to \bin^{n-1}$ and $q\colon \bin^r\to\bin$ are degree-$d$ polynomials acting on the same set of input bits, corresponding to the first $n-1$ output bits and last output bit of $P$, respectively.

Suppose for contradiction that $P(\Ub^r)$ exactly generates $\Db$.
Viola's argument considers the random variable $\Mb$ over $\bin^{n-1}$ defined by sampling $\ub \sim \Ub^r$ and outputting $Q(\ub)$ if $q(\ub) = 1$, and otherwise $n-1$ uniformly random bits independent of $\Ub^r$, which we denote by $\tilde{\Ub}^{n-1}$.
Written suggestively similar to a polynomial, we have
\begin{equation}\label{eq:M_def1}
    \Mb = q(\Ub^r)Q(\Ub^r) + (1-q(\Ub^r)) \tilde{\Ub}^{n-1},
\end{equation}
or equivalently by our assumption,
\begin{equation}\label{eq:M_def2}
    \Mb = \ext(\Ub^{n-1})\Ub^{n-1} + (1 - \ext(\Ub^{n-1})) \tilde{\Ub}^{n-1}.
\end{equation}

Now consider the effect of applying $\ext$ to $\Mb$.
By \Cref{eq:M_def1}, $\Mb$ can be generated by a degree-$(2d)$ $\Fbb_2$-polynomial source with $n - O(1)$ bits of min-entropy (recall $Q(\Ub^r)$ is assumed to be uniform), so $\Pr[\ext(\Mb)=1] \le  \frac{1}{2} + \eps$.
By contrast, whenever the sample $\ub\sim\Ub^{n-1}$ satisfies $\ext(\ub) = 1$, we must have $\ext(\Mb)=1$ by \Cref{eq:M_def2}, so
\begin{align*}
    \Pr[\ext(\Mb)=1] &= \Pr\sbra{\ext(\Mb)=1 \mid \ext(\Ub^{n-1}) = 1}\cdot \Pr[\ext(\Ub^{n-1}) = 1] \\
    & \qquad + \Pr\sbra{\ext(\Mb)=1 \mid \ext(\Ub^{n-1}) = 0}\cdot \Pr[\ext(\Ub^{n-1}) = 0] \\
    &= 1 \cdot \Pr[\ext(\Ub^{n-1}) = 1]  + \Pr[\ext(\tilde{\Ub}^{n-1}) = 1] \cdot \Pr[\ext(\Ub^{n-1}) = 0] \\
    &\approx \frac{1}{2} + \frac{1}{4} \gg \frac{1}{2} + \eps,
\end{align*}
a contradiction.
That is, no low-degree polynomial can exactly sample $\Db$.
In fact, a more careful analysis forbids $P$ from sampling $\Db$ to distance better than $2^{-\Omega(k)}$ (which is $2^{-\Omega(n)}$ for known explicit constructions of such extractors \cite{chattopadhyay2024extractors}).

\subsection{A Constant Lower Bound}

It is initially unclear how to strengthen the previous argument to obtain even some constant distance lower bound.
The primary issue is that even if $\Db$ has large min-entropy, it might be reasonably close to some distribution $P(\Ub^n)$ with small min-entropy, so we do not have any guarantees on the output of the extractor on $\Mb$.

\paragraph{Min-Entropy Polarization.}
One approach that was successful in the analysis of distributions generated by low-depth circuits is \emph{min-entropy polarization} \cite{viola2020sampling}.
Using random restrictions and hypercontractivity, Viola proved that any distribution $C(\Ub^r)$ produced by a small circuit $C\colon \bin^r \to \bin^n$ could be approximated by a not-too-large collection of restrictions of $C$ whose output was either constant or had large min-entropy, in which case an extractor $\ext$ could be meaningfully applied (see \Cref{ssec:circuit} for additional details).
The event witnessing the total variation distance can then be defined as the union of a small set of outputs, corresponding to restrictions under which the circuit is constant, and pairs $(x,b)$ where $\ext(x) \ne b$.

Unfortunately, one cannot always polarize the output distribution of an arbitrary sampler in this fashion; consider the following example from \cite{viola2020sampling}.

\begin{example}
    Let $\Gb$ be the distribution over $\bin^n$ obtained by sampling a random string $\ub \sim \Ub^n$, and outputting $\ub$ if $|\ub| \bmod{2} = 1$ and otherwise outputting $0^n$.
    This sampler cannot have its min-entropy polarized via restrictions for the following reason. For any restriction which leaves some variable free, the min-entropy remains $1$. This means that every restriction in a polarizing collection must fix every variable, but then $\Omega(2^n)$ many such restrictions are needed to approximate the distribution.
\end{example}

Note, however, that this sampling procedure \emph{can} be implemented by the degree-2 $\Fbb_2$-polynomial map
\[
    P(y) = \pbra{\sum_{i=1}^n y_i}\cdot (y_1, \dots, y_n) + \pbra{1-\sum_{i=1}^n y_i}\cdot 0^n
\]
to express $\Gb$ as $P(\Ub^n)$.
Hence, we are not able to apply this framework for our purposes.

\paragraph{Discarding Heavy Points.}
An alternative approach, still somewhat in the spirit of \cite{viola2020sampling}, was taken by Chattopadhyay, Goodman, and Zuckerman \cite{chattopadhyay2022space} when considering communication sources (see \Cref{ssec:communication} for a formal definition).
They also defined their witness event to be the union of a small ``bad'' set and input-output pairs which contradict a particular extractor's definition.
Unlike Viola, however, they define their bad set to be strings which are assigned substantial probability by the sampler (excluding the last bit).
Since there cannot be many ``heavy'' strings, this set is indeed small.

The downside to this approach is that the extractor one requires does not necessarily correspond to an extractor for the class of sources being considered.
In particular, the argument requires an extractor that works on the source after conditioning on membership in some set.
For communication sources, one can reduce to the case of being able to apply a two-source extractor, but it is unclear how to generically perform a similar reduction to known extractors in the case of other sources, such as those generated by low-degree polynomials.

\paragraph{Robust Extractors.}
To overcome the previous two obstacles, we ask for more from the \emph{extractor} rather than from the \emph{source}.
Recall that the issue with trying to strengthen the argument from \Cref{ssec:nonzero_LB} is that the generated distribution $P(\Ub^r)=(Q(\Ub^r), q(\Ub^r))$ might assign too much mass to a small set of ``bad'' points, so the auxiliary distribution of $\Mb$ does not have enough min-entropy to apply an extractor.
Our solution is to simply ask for an extractor whose soundness holds on the light (i.e., low probability) points, rather than the entire domain.
In other words, we seek an extractor which is robust to a small number of points violating the min-entropy constraint.

To be a bit more precise, suppose that the $(\eps,k)$-extractor $\ext\colon \bin^{n-1} \to \bin$ we have been considering had the property that every degree-$(2d)$ $\Fbb_2$-polynomial source $\Xb$ satisfies
\begin{equation}\label{eq:og_robust_ext_def_overview}
    \Pr[\Xb \in L \text{ and } \ext(\Xb) = 1] \le \frac{1}{2} + \eps,
\end{equation}
where $L = \{x \in \bin^{n-1} : \Pr[\Xb = x] \leq 2^{-k}\}$ is the set of light points.\footnote{There are other ways to capture a similar notion of robustness, but this definition seems to be particularly advantageous; see the discussion in \Cref{ssec:ext_and_rext}.}
(Recall that a ``traditional'' extractor is only guaranteed to satisfy \Cref{eq:og_robust_ext_def_overview} when $L = \bin^{n-1}$.)
A priori, it is not obvious that these objects even exist, but it turns out that a straightforward modification to the extractor construction of \cite{chattopadhyay2024extractors} yields such an object.

In general, we believe that essentially any natural class of distributions should have extractors with a similar guarantee, and much of the present work is devoted to showing that existing extractor constructions (or simple modifications of them) are robust.
In fact, one can interpret the aforementioned min-entropy polarization results in \cite{viola2020sampling} as saying that any extractor for circuit sources enjoys the robustness property (see \Cref{ssec:circuit}).

Now that we have a \emph{robust extractor}, we can run a version of our earlier argument.
Once again, let $\Mb$ be the random variable over $\bin^{n-1}$ defined by sampling $\ub \sim \Ub^r$ and outputting $Q(\ub)$ if $q(\ub) = 1$, and otherwise $n-1$ random bits using fresh randomness.
We also define $\tilde{\Mb}$ similarly, but with $\Db = (\Ub^{n-1}, \ext(\Ub^{n-1}))$ in place of $P(\Ub^r)$.
By \Cref{eq:og_robust_ext_def_overview}, we morally have that $\ext(\Mb)$ is near balanced on its light points, whereas by construction, $\ext(\tilde{\Mb})$ is likely to output 1 on such points.
Working out the details (see \Cref{ssec:main_result}), one finds that $P(\Ub^r)$ must be $\Omega(1)$-far from $\Db$, as desired.

\subsection{A $1-o(1)$ Lower Bound}

It remains to amplify the total variation distance lower bound from constant to $1-o(1)$.
Toward this end, we will take many independent copies of the distribution above, and prove that generating this product of distributions is much harder than generating any individual distribution.
We note that such \emph{direct product theorems} have been examined previously in the sampling context \cite{chattopadhyay2022space, GKM+26}.

Proceeding more formally, we redefine the distribution $\Db$ to be
\[
    \Db = (\Ub_1, \Ub_2, \dots, \Ub_t, \ext(\Ub_1), \ext(\Ub_2), \dots, \ext(\Ub_t)),
\]
where the $\Ub_i$'s are independent copies of $\Ub^{n-1}$ and $t \ge 1$ is an integer to be chosen later.
Observe that now $\Db$ is over $\bin^N$ for $N \coloneqq tn$, so we must redefine our polynomial map $P$ to also be on $N$ output bits.
We emulate much of our previous analysis with the auxiliary random variable $\tilde{\Mb}$ generalized to output $\ub_i$ if $\ext(\ub_i) = 1$ for any $i \in \cbra{1,\dots, t}$ (and $\Mb$ similarly generalized with respect to $P(\Ub^r)$).
This greatly increases the probability that $\ext(\tilde{\Mb}) = 1$, which further improves our distance bound.

Unfortunately, the approach as written does not provide distance approaching 1.
The problem can essentially be traced back to a loss of $\Pr[\Xb \in L \text{ and } \ext(\Xb) = 1]$ coming from our test event.
Since the set of low probability points $L$ might be the entire domain, we cannot upper bound this quantity by anything appreciably better than 1/2 without violating the extractor guarantee.
In order to make further progress, we extend our extractor to have multiple output bits.

This multi-output extractor $\ext\colon \bin^{n-1} \to \bin^m$ is defined similarly to the single-output version (\Cref{def:extractor}), only now with the pseudorandomness condition being a bound on the total variation distance $\tvdist{\ext(\Xb) - \Ub^m} \le \eps$.
One can again obtain a robust version (with $m = \Omega(\log\log n)$) by modifying the construction of \cite{chattopadhyay2024extractors}; we defer the formal definition to \Cref{ssec:ext_and_rext}.

The upshot is that now we can replace the bottleneck of the previous argument with an analysis of $\Pr[\Xb \in L \text{ and } \ext(\Xb) = 0^m]$.
This, in turn, can be more tightly bounded by $2^{-m}+\eps$.
In our construction, $\eps$ will tend to 0 as $m$ grows,\footnote{We had previously defined $\eps$ to be a small constant, but our construction of robust extractors allows for this smaller setting.}
so working out the calculations, we obtain a bound roughly of the form 
\[
    1 - \exp(-2^{-m}\cdot t) - 2^{-\Omega(n)}\cdot \poly(t)
\]
for sufficiently large $n$.
Setting $t$ to be slightly larger than $2^{m}$, we can finally obtain our desired $1 - o(1)$ bound.
The full details can be found in the subsequent subsection.

\paragraph{Isolators.}

Before formalizing our discussion into a theorem, it is worth reflecting on exactly what properties of the robust extractor we used.
(These properties may be further illuminated by consulting the proof of \Cref{thm:main}.)
In analyzing our target distribution $\Db$, all that we truly needed is for some output $y \in \bin^m$ to be assigned substantial mass by $\ext(\Ub^{n-1})$.
By contrast, our analysis of the generated distribution $P(\Ub^r)$ only required that the extractor did not map too many of the light points to the same output $y$.
In other words, our argument does not require the full power of robust extractors.
We distill the properties we need into the following weaker object, which we call an \emph{isolator}, since it in some sense controls the isolated light points.

\begin{restatable}[Isolator]{definition}{defisolator}\label{def:isolator}
    A function $\iso\colon \bin^n \to \bin$ is an \emph{$(\alpha, \beta, k)$-isolator} for a class $\Xcal$ of distributions over $\bin^n$ if
    \begin{enumerate}
        \item $\Pr[\iso(\Ub^n) = 1] \ge \alpha$, and

        \item Every source $\Xb \in \Xcal$ satisfies $\Pr_{x \sim \Xb}[x \in L \text{ and } \iso(x) = 1] \le \beta$, where $L = \{x \in \bin^n : \Pr[\Xb = x] \leq 2^{-k}\}$.
    \end{enumerate}
\end{restatable}

It is perhaps worth highlighting that although all our constructions of isolators essentially work by ``robustifying'' existing extractor constructions, this is not strictly necessary to satisfy \Cref{def:isolator}.
It is plausible one could explicitly construct an isolator for $\ac[\oplus]$ sources, for example, without first explicitly constructing seemingly elusive extractors for that class.

\subsection{Main Result}\label{ssec:main_result}

We conclude \Cref{sec:main_pf} by stating our main result, \Cref{thm:main}, and its proof.
Below, the following notation will be convenient for generalizing the random variable $\Mb$ considered above.

\begin{definition}[Randomized Addressing]
    For positive integers $t$ and $n$, we define the randomized function $\add\colon (\bin^n)^t \times \bin^t \to \bin^n$ as
    \[
        \add(A_1, A_2, \dots, A_t, b_1, b_2, \dots, b_t) = 
        \begin{cases}
            A_1 & \text{ if } b_1 = 1, \\
            A_2 & \text{ if } b_1 = 0, b_2 = 1, \\
            A_3 & \text{ if } b_1 = b_2 = 0, b_3 = 1, \\
            &\vdots \\
            A_t & \text{ if } b_1 = b_2 = \dots =b_{t-1} = 0, b_t = 1, \\
            \Ub^n & \text{ if } b_1 = b_2 = \cdots = b_t = 0.
        \end{cases}
    \]
\end{definition}
Though the parameters $n$ and $t$ will mostly be clear from context, we will sometimes write $\add_{n, t}$ for clarity.

We can now finally present our main result.

\begin{theorem}\label{thm:main}
    Let $\alpha, \beta \in [0,1]$ and $t,k,n$ be positive integers where $k \le n-1$.
    Let $\Xcal$ be a class of distributions over $\bin^{t(n+1)}$, and let $\Ycal$ be the class of distributions over $\bin^n$ of the form $\add_{n, t}(\Xb)$ for some $\Xb \in \Xcal$.
    Suppose $\iso\colon \bin^n \to \bin$ is an $(\alpha, \beta, k)$-isolator for $\Ycal$, and define the distribution
    \[
        \Db = (\Ub_1, \Ub_2, \dots, \Ub_t, \iso(\Ub_1), \iso(\Ub_2), \dots, \iso(\Ub_t)),
    \]
    where $\Ub_1, \dots, \Ub_t$ are independent copies of $\Ub^n$.
    Then every source $\bf{X}\in \Xcal$ satisfies
    \[
        \tvdist{\Xb - \Db} \ge 1 - (1-\alpha)^{t+1} - 2^{-(n-k)}\pbra{2t^3 + 1} - \beta.
    \]
    In particular, $\tvdist{\Xb - (\Ub^n, \iso(\Ub^n))} \ge 2\alpha - \alpha^2 - 2^{-(n-k-2)} - \beta$.
\end{theorem}

\begin{remark}
    Depending on the sampling model and particular goal, the precise ordering of the coordinates in a target distribution may affect whether or not that distribution can be sampled (e.g., \cite{viola2012extractors}).
    In the case of \Cref{thm:main}, the proof goes through for any permutation of the coordinates of $\add$ and $\Db$, although the order may affect whether $\add(\Xb)$ lies in a particular class.
\end{remark}

\begin{proof}[Proof of \Cref{thm:main}]
    Let $\Xb$ be an arbitrary source in $\Xcal$.
    Viewing $\Xb$ as $(\Xb_1, \Xb_2, \dots, \Xb_t, \xb_1,\xb_2,\dots, \xb_t)$ with the $\Xb_i$'s and $\xb_j$'s over $\bin^n$ and $\bin$, respectively, we may apply $\add_{n, t}$ to $\Xb$ to obtain the source $\Yb = \add(\Xb) \in \Ycal$.
    Next, define the sets of light points in $\Xb_i$ for $i = 1,2,\dots, t$ and $\Yb$ by
    \[
        L_{\Xb_i} = \{x \in \bin^{n} : \Pr[\Xb_i = x] \leq 2^{-(\log(t)+k+1)}\} \quad\text{and}\quad L_\Yb = \{y \in \bin^{n} : \Pr[\Yb = y] \leq 2^{-k}\}.
    \]
    Observe that $\bigcap_i L_{\Xb_i} \subseteq L_\Yb$, since any string $z \in \bigcap_i L_{\Xb_i}$ satisfies
    \begin{align*}
        \Pr\sbra{\Yb = z} = \Pr_\Xb\sbra{\add(\Xb) = z} &\le \Pr\sbra{\Ub^n = z} + \sum_{i=1}^t \Pr\sbra{\Xb_i = z} \tag{by union bound} \\
        &\le 2^{-n} + t\cdot 2^{-(\log(t) + k + 1)} \le 2^{-k}.
    \end{align*}
    Thus, the isolator property of $\iso$ guarantees that
    \begin{align}\label{eq:thm:main_extractor_property}
        \Ecal(\Xb) \coloneqq & \Pr_\Xb \sbra{\bigcap_i \pbra{\Xb_i \in \bigcap_j L_{\Xb_j}} \text{ and } \iso(\add(\Xb)) = 1} \notag \\
        \le & \: \Pr_\Xb \sbra{\add(\Xb) \in L_\Yb \text{ and } \iso(\add(\Xb)) = 1} + \Pr_\Xb \sbra{\add(\Xb) \not\in L_\Yb \mid \bigcap_i \pbra{\Xb_i \in \bigcap_j L_{\Xb_j}}} \notag \\
        \le & \:\beta + \Pr\sbra{\Ub^n \not\in L_\Yb} \le \beta + 2^{-(n-k)},
    \end{align}
    where the final inequality follows from the fact that $\Yb$ can assign $2^{-k}$ probability mass to at most $2^k$ points.
    
    We will choose $\Ecal$ to be the test witnessing the claimed TVD between $\Xb$ and $\Db$, although it will be slightly more convenient in this part of the analysis to consider the complement event.
    In a similar fashion to $\Xb$, we view $\Db$ as $(\Db_1, \Db_2, \dots, \Db_t, \db_1, \db_2, \dots, \db_t$), where recall each $\Db_i$ is an independent copy of the uniform distribution $\Ub^n$.
    By the union bound, we have
    \begin{align}\label{eq:thm:main_union_bound}
        \Pr_{\Db} \sbra{\bigcup_i\pbra{\Db_i \not\in \bigcap_j L_{\Xb_j}} \text{ or } \iso(\add(\Db)) = 0} &\le \sum_i \Pr_\Db \sbra{\Db_i \not\in \bigcap_j L_{\Xb_j}} + \Pr_\Db[\iso(\add(\Db)) = 0] \notag \\
        &\le \sum_{i,j} \Pr_\Db \sbra{\Db_i \not\in L_{\Xb_j}} + \Pr_\Db[\iso(\add(\Db)) = 0].
    \end{align}
    Note that the first term is at most $t^2\cdot 2^{-(n-(\log(t)+k+1))}$, since no $\Xb_j$ can assign $2^{-(\log(t)+k+1)}$ probability mass to more than $2^{\log(t)+k+1}$ points. 
    To address the second term, let us consider the effect of applying $\add$ to $\Db$. Whenever any $\Db_i$ satisfies $\iso(\Db_i) = 1$, we have $\iso(\add(\Db)) = \iso(\Db_i) = 1$.
    Otherwise, $\iso(\Db_i) = 0$ for all $i$, and we have $\iso(\add(\Db)) = \iso(\Ub^n)$.
    Thus,
    \begin{align*}
        \Pr[\iso(\add(\Db)) = 0] &= \Pr_\Db \sbra{\iso(\Ub^n) = 0 \text{ and } \bigcap_i \pbra{\iso(\Db_i) = 0}} \\
        &= \Pr[\iso(\Ub^n) = 0]^{t+1} \le (1-\alpha)^{t+1}.
    \end{align*}
    Plugging these two bounds back into \Cref{eq:thm:main_union_bound} yields
    \[
        \Pr_{\Db} \sbra{\bigcap_i\pbra{\Db_i \in \bigcap_j L_{\Xb_j}} \text{ and } \iso(\add(\Db)) = 1} \ge 1 - (1-\alpha)^{t+1} - t^2\cdot 2^{-(n-(\log(t)+k+1))}.
    \]
    Recalling the upper bound from \Cref{eq:thm:main_extractor_property}, we conclude
    \[
        \tvdist{\Xb - \Db} \ge 1 - (1-\alpha)^{t+1} - 2^{-(n-k)}\pbra{2t^3 + 1} - \beta. \qedhere
    \]
\end{proof}

\section{Preliminaries}\label{sec:prelim}

We now briefly review some notation and formal definitions used throughout the work.

\subsection{The Basics}

For a positive integer $n$, we use $[n]$ to denote the set $\cbra{1,2,\dots, n}$.
For a binary string $x$, we use $|x|$ to denote the Hamming weight of $x$.
All logarithms given in the paper are base 2.
For two real numbers $a$ and $b$, we write $a \pm b$ to denote a value in the range $[a-b, a+b]$.
The indicator function is denoted by $\indicator(\cdot)$.
The notation $\binom{n}{\le k}$ is shorthand for $\sum_{i=0}^k \binom{n}{i}$.
We use $\Fbb_2$ to denote the finite field of two elements; we often identify it with $\bin$.
The concatenation of two strings $x$ and $y$ is denoted $x \circ y$.

\paragraph{Asymptotics.}
We use the standard $\Omega(\cdot), O(\cdot), \Theta(\cdot)$ asymptotic notation to hide universal positive constants, although we will sometimes use $\gg$ and $\ll$ in more informal contexts.
Occasionally, we will use subscripts to indicate an unspecified dependence on a particular parameter (e.g., $\Omega_d(n)$).
Additionally, we write $o_t(1)$ to denote a positive quantity tending to 0 as $t$ tends to infinity; we often omit the subscript when $t$ is clear from context.
The shorthand $\poly(n)$ corresponds to a polynomial in $n$ of some fixed, but unspecified, degree.

\paragraph{Probability.}
We endeavor to use bold capital letters to denote probability distributions, and use bold lowercase letters to denote randomly drawn samples.
That is, $\xb \sim \Xb$ denotes a sample $\xb$ drawn from the distribution $\Xb$.
Oftentimes, we will refer to a distribution as a \emph{source} if it is being fed into an extractor-like object.
We reserve $\Ub$ for the uniform distribution over $\bin$.
We use calligraphic letters, such as $\Xcal$, for classes of distributions.
For an event $\Ecal$, we define $\Xb(\Ecal)$ to be the probability mass assigned to $\Ecal$ by $\Xb$.
For a function $f$, we use $f(\Xb)$ to denote the output distribution of $f(\xb)$ on randomly drawn $\xb\sim \Xb$.
The \emph{min-entropy} of a distribution $\Xb$, denoted $H_\infty(\Xb)$, is given by $-\log \max_{x \in \supp{\Xb}} \Pr[\Xb=x]$, where the support $\supp{\Xb} = \cbra{x : \Pr[\Xb = x] > 0}$.

Given a distribution $\Xb$ and positive integer $t$, we use $\Xb^t$ to denote the $t$-fold product distribution $\Xb \times \cdots \times \Xb$.
If $s$ is a finite set, we write $\Xb^s$ to emphasize that the coordinates of $\Xb^{|s|}$ are indexed by $s$.
We refer to $\Xb$ as a mixture if it can be written as a convex combination of other distributions.
That is, there exists $c_1, \dots, c_k \in [0,1]$ and distributions $\Xb_1, \dots, \Xb_k$ such that $\Xb(\Ecal) = \sum_{i=1}^k c_i \cdot \Xb_i(\Ecal)$ for every event $\Ecal$.
Occasionally, we write this more concisely as $\Xb = \sum_{i=1}^k c_i \Xb_i$.

We measure the similarity of two (discrete) distributions $\Pb$ and $\Qb$ by the \emph{total variation (TV) distance} 
\[
    \tvdist{\Pb - \Qb} = \max_{\text{event } \Ecal} \Pb(\Ecal) - \Qb(\Ecal) = \frac{1}{2}\sum_x \abs{\Pb(x) - \Qb(x)}.
\]
We say $\Pb$ is \emph{$\eps$-close} to $\Qb$ if $\tvdist{\Pb - \Qb} \le \eps$, and \emph{$\eps$-far} otherwise.

\paragraph{Explicit Constructions.}
The focus of this work is on \emph{explicitly} constructing distributions with certain properties, by which we mean $\Db$ can be sampled in $\poly(n)$ time.
Note that as in the computational world, it is straightforward to non-constructively prove strong hardness results via a counting argument.
More precisely, for any class $\Xcal$ of distributions over $\bin^n$ of size $2^{2^{cn}}$ for a constant $c < 1$, there exists a uniform distribution $\Db$ with distance $1 - 2^{-\Omega(n)}$ from every distribution in $\Xcal$ (see \Cref{clm:non_constructive_distance}).

There is a small subtlety, however, in that the sampling models we consider have unbounded input length, so one cannot naively apply the above claim.
Fortunately, the classes of interest can typically be approximated by a subset of the class where the input length is bounded (e.g., \cite{chattopadhyay2022space}), which allows the argument to go through.
We are unaware of any work formally describing these counting arguments giving hard distributions for low-degree polynomial sources and $\ac[\oplus]$ sources, so we record the details in \Cref{app:non-constructive}.

\subsection{Extractors, Robust Extractors, and Isolators}\label{ssec:ext_and_rext}

There are three pseudorandom objects at the core of our work: extractors, robust extractors, and isolators.
We state the formal definitions of the first two below.

\begin{definition}[(Robust) Extractor]
\label{def:ext_and_robust_ext}
    A function $\ext\colon \bin^n \to \bin^m$ is an \emph{$(\eps,k)$-extractor} for a class $\Xcal$ of distributions over $\bin^n$ if for every source $\Xb \in \Xcal$ with min-entropy at least $k$, we have
    \[
        \tvdist{\ext(\Xb) - \Ub^m} \le \eps.
    \]
    If additionally, there exists some $z\in \bin^m$ such that every source $\Xb \in \Xcal$ satisfies
    \[
        \Pr[\Xb \in L \text{ and } \ext(\Xb) = z] \le \frac{1}{2^m} + \delta,
    \]
    where $L = \{x \in \bin^{n} : \Pr[\Xb = x] \leq 2^{-k}\}$, then we call $\ext$ an \emph{$(\eps,\delta,k)$-robust extractor}.
\end{definition}

It is worth mentioning that there are other ways to formalize this notion of robustness.
For example, one could ask for an extractor which works for a source $\Xb$, as long as $\Xb$ has TV distance no more than some parameter $\gamma$ to a high min-entropy source $\Xb'$ (i.e., an extractor for sources of high smooth min-entropy \cite{RW04}).
Assuming such extractors existed for the classes of distributions we consider, it would be possible to obtain a constant distance lower bound as in the second part of \Cref{thm:main}.

Unfortunately, it does not seem feasible to get distance approaching 1.
Fix some choice of $\gamma$ and a sampler $f(\Ub^r)$ for the hard distribution $\Db$ considered in \Cref{thm:main}.
If $f(\Ub^r)$ is $\gamma$-far from $\Db$, we trivially have a lower bound of $\gamma$, so assume this is not the case.
Note that we cannot expect a ``smooth extractor'' to have a better upper bound probability guarantee than $2^{-m} + \gamma$, since it has to apply to distributions which are a point mass with probability $\gamma$ and uniform otherwise.
Tracing through the remainder of the argument, one finds that $f(\Ub^r)$ has distance roughly $1 - \eps - \gamma$ from $\Db$ for some small $\eps$.
In other words, we can only guarantee a bound around $\min(\gamma, 1-\eps -\gamma) \ll 1$.

One of the benefits of \Cref{def:ext_and_robust_ext} is that we can obtain tighter distance guarantees, because we do not need to consider the behavior of the extractor on heavy points.
Hence, we can avoid the above issue.
Returning back to our selected formalizations, we restate the definition of an isolator for the reader's convenience.

\defisolator*

The role of robust extractors in the present work is primarily as a convenient device for constructing isolators.

\begin{fact}\label{fct:iso_from_robust_ext}
    Let $\Xcal$ be a class of distributions over $\bin^n$ which includes the uniform distribution $\Ub^n$.
    If $\rext\colon \bin^n \to \bin^m$ is an $(\eps,\delta,k)$-robust extractor for $\Xcal$, then the function $\iso(x) \coloneqq \indicator(\rext(x) = z)$ is a $(2^{-m} - \eps, 2^{-m} + \delta, k)$-isolator for $\Xcal$ (where $z$ is the same string as in \Cref{def:ext_and_robust_ext}).
\end{fact}

Of course, one still has to construct a robust extractor to apply \Cref{fct:iso_from_robust_ext}.
Typically, this is not much more difficult than constructing a traditional extractor, and in the following section, we illustrate such a modification with a number of natural examples.
We note that in certain cases, the robustness can even be obtained for free.

\begin{claim}\label{clm:unif_sources}
    Let $\Xcal$ be a class of distributions over $\bin^n$ such that each $\Xb \in \Xcal$ is the uniform distribution on some set $S_\Xb \subseteq \bin^n$ (i.e., $\Xcal$ is a family of \emph{flat} sources).
    If $\ext\colon\bin^n \to \bin^m$ is an $(\eps, k)$-extractor for $\Xcal$, then $\ext$ is also an $(\eps, \eps, k)$-robust extractor for $\Xcal$.
\end{claim}
\begin{proof}
    Consider any source $\Xb \in \Xcal$, which is uniform over a set $S_{\Xb}$.
    If $|S_{\Xb}| \ge 2^k$, then $\Xb$ has min-entropy at least $k$ and the set
    \[
        L \coloneqq \{x \in \bin^n : \Pr[\Xb = x] \leq 2^{-k}\} = \bin^n.
    \]
    Thus, the extractor guarantee implies 
    \[
        \Pr_{x \sim \Xb}[x \in L \text{ and } \ext(x) = 0^m] = \Pr_{x \sim \Xb}[\ext(x) = 0^m] \leq 2^{-m} + \eps.
    \]
    Otherwise $|S_{\Xb}| < 2^k$, and $\Pr[\Xb = x] = 0$ for every $x \in L$.
    Hence,
    \[
        \Pr_{x \sim \Xb}[x \in L \text{ and } \ext(x) = 0^m] \le \Pr_{x \sim \Xb}[x \in L] = 0. \qedhere
    \]
\end{proof}

\section{Hard-to-Sample Distributions for Specific Sources}\label{sec:hard_dist_instantiation}

In this section, we instantiate \Cref{thm:main} for a number of commonly studied sources.
Most of the results we obtain are already known (with the notable exception of polynomial sources), but we reprove them in our framework as evidence that typical constructions of extractors can be easily adapted into ones for robust extractors or isolators.

\subsection{Polynomial Sources}\label{ssec:poly}

We begin with polynomial sources, as achieving strong sampling lower bounds in this setting is one of the paper's main contributions.

\begin{definition}[Polynomial Source]\label{def:poly_source}
A degree-$d$ \emph{polynomial source} $\Xb$ is defined by a polynomial map $P\colon\Fbb_2^m \rightarrow \Fbb_2^n$, where $P = (p_1, p_2, \dots, p_n)$ and each $p_i$ is an $\Fbb_2$-polynomial of degree at most $d$, such that $\Xb = P(\Ub^m)$.
\end{definition}

Prior to our work, the best explicit distribution was only known to have distance $2^{-\Omega(n)}$ from any low-degree polynomial map; this is achieved by combining an argument of \cite{viola2014extractors} (described in \Cref{ssec:nonzero_LB}) with an extractor construction of \cite{chattopadhyay2024extractors}.
We improve this distance bound to $1-o(1)$.

\begin{restatable}{theorem}{thmharddistpoly}\label{thm:hard_dist_poly}
    There exists a constant $\delta > 0$ such that the following holds.
    Let $N$ and $\Delta$ be positive integers satisfying $\Delta \leq \delta \log \log N$. 
    There exist positive integers $n, t, d$ with $(n+1)t \leq N$ and an isolator $\iso\colon \bin^n \to \bin$ with suitable parameters for the class of polynomial sources on $\bin^n$ of degree $d$ such that the following holds.

    Define the distribution 
    \[
        \Db=(\Ub_1, \Ub_2, \dots, \Ub_t, \iso(\Ub_1), \iso(\Ub_2),\dots, \iso(\Ub_t)),
    \]
    where $\Ub_1, \Ub_2, \dots, \Ub_t$ are independent copies of $\Ub^n$.
    Let $\Xcal$ be the class of all polynomial sources $\Xb$ on $\bin^{(n+1)t}$ of degree $\Delta$. Then for any $\Xb \in \Xcal$,
    \[
        \tvdist{\Db - \Xb} \geq 1 - O\left(\frac{\Delta}{\log \log N} \log \left(\frac{\log \log N}{\Delta}\right)\right).
    \]
    Moreover, the distribution $\Db$ can be sampled in $\poly(N)$ time.

    Additionally, there exists an isolator $\tilde{\iso} \colon \bin^n \to \bin$ (with possibly different parameters) such that the distribution $\tilde{\Db} = (\Ub^n, \tilde{\iso}(\Ub^n))$ satisfies $\tvdist{\tilde{\Db} - \Yb} \geq 1/4 - n^{-\Omega(1)}$ for all degree-$\Delta$ polynomial sources $\Yb$ on $\bin^{n+1}$.
    $\tilde{\Db}$, too, can be sampled in $\poly(N)$ time.
\end{restatable}

We prove \Cref{thm:hard_dist_poly} by modifying the construction of an explicit extractor from \cite{chattopadhyay2024extractors} to create an isolator $\iso$ for low-degree polynomial sources, and invoke \Cref{thm:main} using $\iso$.

The extractor from \cite{chattopadhyay2024extractors} is constructed by brute-forcing over many possible extractors on a small number of input bits.
The existence of such an extractor is guaranteed by the probabilistic method (as is standard), but what makes their argument not straightforward is that the number of polynomial sources is not bounded, since a polynomial source can have any number of input bits. 
Their key ingredient is an input reduction technique \cite[Theorem 4.1]{chattopadhyay2024extractors}, which shows that it is enough to consider only sources where the number of input bits is at most a constant factor times the min-entropy.

We observe below that a variant of their input reduction argument also works for isolators. 
We first prove a variant of their simple entropy smoothing claim \cite[Claim 4.3]{chattopadhyay2024extractors}.

\begin{claim}\label{clm:smoothing}
    For any random variable $\Xb$ over $\bin^n$ and positive integer $k$, there exists a function $S\colon\bin^n \rightarrow \bin^{k+1}$ with the following property. For every $x \in \bin^n$ such that $\Pr[\Xb = x] \leq 2^{-k}$, we have $\Pr[S(\Xb) = S(x)] \leq 2^{-k}$.
\end{claim}
\begin{proof}
    We perform the following merging operation. Start with each string in $\bin^n$ in its own bucket. If there are two buckets whose combined probability mass under $\Xb$ is at most $2^{-k}$, then merge the two buckets. This operation is repeated for as long as possible. (This process necessarily terminates since the number of buckets decreases in each step.)

    At the end, there can be at most one bucket whose corresponding probability is at most $2^{-(k+1)}$, since otherwise we could merge two such buckets. This implies that there are at most $2^{k+1}$ buckets in total, and we can define $S$ to simply map each string $x \in \bin^n$ to the bucket containing it (where we associate each bucket with a distinct string in $\bin^{k+1}$).

    To verify the desired property of $S$, first observe that since any bucket containing two or more elements is the result of some merge operation, it must have total probability at most $2^{-k}$. 
    This handles all $x$ such that $|S(x)| \geq 2$. On the other hand, if $\Pr[\Xb = x] \leq 2^{-k}$ and $|S(x)| = 1$, then clearly $\Pr[S(\Xb) = S(x)] = \Pr[\Xb = x] \leq 2^{-k}$.
\end{proof}

We also need the following general lemma, which is implicit \cite{chattopadhyay2024extractors}. 
We sketch the proof for completeness.

\begin{restatable}{lemma}{leminputredux}\label{lem:input_redux}
    Let $f \colon \Fbb_2^r \rightarrow \Fbb_2^n$ be a function and $0 < \eps < 1/4$. Let $\ell = \ceilbra{n + 3\log(1/\eps)}$. 
    If $r > \ell$, there exist $A \in \Fbb_2^{r \times \ell}$ and $b \in \Fbb_2^r$ such that if we define $h \colon \Fbb_2^\ell \rightarrow \Fbb_2^n$ by $h(x) = f(Ax + b)$, we have  $\tvdist{f(\Ub^r) - h(\Ub^\ell)} \leq 2\eps$.
\end{restatable}

\begin{proof}
    Lemma 4.2 in \cite{chattopadhyay2024extractors} gives a full rank\footnote{The conference version of \cite{chattopadhyay2024extractors} does not explicit state that $M$ is full rank, but this can be found in the updated arXiv version: \url{https://arxiv.org/abs/2309.11019}.} linear map $M \colon\Fbb_2^r \rightarrow \Fbb_2^{r-\ell}$ such that
    \begin{equation}\label{eq:lem:f2_input_red:1}
        \tvdist{f(\Ub^r) \circ M(\Ub^r) - f(\Ub^r) \circ \Ub^{r-\ell}} \leq 2\eps.
    \end{equation}
    We start by observing that  
    \[
        \tvdist{f(\Ub^r) \circ M(\Ub^r) - f(\Ub^r) \circ \Ub^{r-\ell}} = \E_{v \sim M(\Ub^r)}\Big[\tvdist{(f(\Ub^r)\mid M(\Ub^r) = v) - f(\Ub^r)}\Big],
    \]
    where we have used that $M$ has full rank, so $M(\Ub^r)$ is $\Ub^{r-\ell}$. 
    Combining with \Cref{eq:lem:f2_input_red:1}, there must exist some $v \in \Fbb_2^{r-\ell}$ such that 
    \[
        \tvdist{(f(\Ub^r)\mid M(\Ub^r) = v) - f(\Ub^r)} \leq 2\eps.
    \]
    Now $\left(f(\Ub^r)\mid M(\Ub^r) = v\right)$ can be expressed as $h(\Ub^\ell)$ for a function $h$ on $\ell$ inputs of the desired form $h(x) = f(Ax + b)$, since conditioning on $M(\Ub^r) = v$ (where $M$ has full rank) is equivalent to replacing $r - \ell$ of the $r$ input bits for $f$ by affine functions of the other $\ell$ inputs. 
    Thus, we have $\tvdist{h(\Ub^\ell) - f(\Ub^r)} \leq 2\eps$ as desired.
\end{proof}

We now proceed to the input reduction lemma.

\begin{lemma} \label{lem:f2_input_red}
    Suppose $\iso\colon \Fbb_2^n \rightarrow \Fbb_2$ is an $(\alpha, \beta, k-1)$-isolator for the class of degree-$d$ polynomial sources with at most $4(k+1)$ inputs. Then $\iso$ is also an $(\alpha, \beta+2^{-k}, k)$-isolator for the class of all degree-$d$ polynomial sources.
\end{lemma}
\begin{proof}
    We clearly have $\Pr[\iso(\Ub^n) = 1] \geq \alpha$ since $\iso$ is an $(\alpha, \beta, k-1)$-isolator, so it remains to verify that $\iso$ satisfies the second condition in the definition of an $(\alpha, \beta+2^{-k}, k)$-isolator for the class of all degree-$d$ polynomial sources.

    Let $f\colon \Fbb_2^r \rightarrow \Fbb_2^n$ be a degree-$d$ polynomial map with light points $L_f \coloneqq \{x : \Pr[f(\Ub^r) = x] \leq 1/2^k\}$.
    Let $\ell = 4(k+1)$, and let $S\colon \bin^n \to \bin^{k+1}$ be the map given by \cref{clm:smoothing} applied to $f(\Ub^r)$ and $k$.
    Apply \cref{lem:input_redux} to the function $S(f(\cdot ))$ with $\eps = 2^{-(k+1)}$ to get $A \in \Fbb_2^{r \times \ell}, b \in \Fbb_2^r$ such that 
    \begin{equation}\label{eq:lem:f2_input_red:2}
    \tvdist{S(f(A \cdot \Ub^\ell + b)) - S(f(\Ub^r))} \leq 2^{-k}.
    \end{equation}
    Let $g \colon \Fbb_2^\ell \rightarrow \Fbb_2^n$ be defined by $g(x) = f(Ax + b)$. Observe that $g$ is a degree-$d$ polynomial since we have only substituted linear polynomials for the inputs of $f$.

    Define $L_g = \{x : \Pr[g(\Ub^\ell) = x] \leq 2^{-(k-1)}\}$.
    Observe that $L_f \subseteq L_g$, since \Cref{eq:lem:f2_input_red:2} implies every $x \in L_f$ satisfies
    \begin{align*}
        \Pr[g(\Ub^\ell)=x] \leq \Pr[S(g(\Ub^\ell))=S(x)] \leq \Pr[S(f(\Ub^r))=S(x)] + 2^{-k} \leq 2^{-(k-1)}.
    \end{align*}
    Therefore, we conclude
    \begin{align*}
        \Pr_{\xb \sim f(\Ub^r)}[\xb \in L_f \text{ and } \iso(\xb) = 1] &\leq \Pr_{\xb \sim g(\Ub^\ell)}[\xb \in L_f\text{ and } \iso(\xb) = 1] + 2^{-k} \tag{by \Cref{eq:lem:f2_input_red:2}} \\
        &\leq \Pr_{\xb \sim g(\Ub^\ell)}[\xb \in L_g\text{ and } \iso(\xb) = 1] + 2^{-k} \tag{since $L_f \subseteq L_g$} \\
        &\leq \beta + 2^{-k}.
    \end{align*}
    The last inequality above uses the isolator guarantee for $g(\Ub^{\ell})$. This finishes the proof.
\end{proof}

To apply \Cref{lem:f2_input_red}, we need to construct an isolator for the class of polynomial sources with bounded input length.
The following lemma will allow us to find such an object more efficiently than a naive brute force.
Below, recall that a family $\cbra{f_i\colon \bin^n \to \bin^m}_i$ of $t$-wise uniform hash functions is defined by the property that for any string $(y_1,y_2,\dots,y_t) \in (\bin^m)^t$ and all distinct strings $x_1, x_2, \dots, x_t \in \bin^n$, a function $\mathbf{f}$ chosen uniformly at random from the family satisfies $(\mathbf{f}(x_1), \mathbf{f}(x_2), \dots, \mathbf{f}(x_t)) = (y_1,y_2,\dots,y_t)$ with probability $2^{-mt}$.

\begin{lemma} \label{lem:bdd_f2_tind_rext}
    Let $0 < \alpha < \beta < 1$. Let $k,n,\ell$, and $d$ be positive integers. Set $K = 2^k$ and $N = 2^n$. Let $t \geq 4$ be an even integer. Suppose there exists an integer $m$ such that $p \coloneqq 2^{-m}$ satisfies  $\alpha < p < \beta$ and the following inequality holds:
    \[
    \left(\frac{Npt +t^2}{N^2(p-\alpha)^2 }\right)^{t/2} + 2^{\binom{\ell}{\leq d}n}\left(\frac{K p t + t^2}{K^2(\beta - p)^2}\right)^{t/2} < 1/8.
    \]
    Additionally, let $\Hcal$ be a family of $t$-wise uniform hash functions from $\bin^n$ to $\bin^m$.
    Then there exists a function $h \in \Hcal$ such that the function $g$ defined by $g(x) = \indicator(h(x) = 0^m)$ is an $(\alpha, \beta, k)$-isolator for the class of degree-$d$ polynomial sources with $\ell$ inputs and $n$ outputs.
\end{lemma}
\begin{proof}
    Let $\hb$ be a function drawn uniformly at random from $\Hcal$, and define $\gb\colon\bin^n \rightarrow \bin$ by $\gb(x) = \indicator(\hb(x) = 0^m)$. 
    By the assumption on $\Hcal$, $\{\gb(x)\}_{x \in \bin^n}$ is a collection of $t$-wise independent random variables, and for each $x$, $\Pr[\gb(x) = 1] = 2^{-m}$.
    By using a tail inequality for sums of $t$-wise independent random variables (see \cite[Lemma 2.3]{bellare1994randomness}), we have
    \[
        \Pr_\gb[\Pr[\gb(\Ub^n) =1] \leq \alpha] \leq 8\left(\frac{Npt +t^2}{N^2(p-\alpha)^2 }\right)^{t/2}.
    \]
    This is the probability that $\gb$ fails to satisfy the first condition for being an $(\alpha, \beta, k)$-isolator.
    
    We will now estimate the probability that for some source, the second condition for being an $(\alpha, \beta, k)$-isolator is not satisfied. Fix any source $\Xb$ on $\bin^n$, and define $L = \{x : \Pr[\Xb = x] \leq 2^{-k}\}$. 
    We wish to show that with high probability over $\gb$, $\Pr[\gb(\Xb) = 1 \text{ and } \Xb \in L] \leq \beta$.
    Let $\Zb = \Pr[\gb(\Xb) = 1 \text{ and } \Xb \in L]$ be the random variable of interest. 
    For each $x \in L$, define the indicator random variable $\Zb_x = \indicator(\gb(x)= 1)$, each of which is a Bernoulli random variable with probability $p$ of being $1$. 
    Then $\Zb = \sum_{x \in L} \Pr[\Xb = x] \Zb_x$ and $\E[\Zb] = \Pr[\Xb \in L] \cdot p \leq p$. 
    
    By a tail inequality for $t$-wise independent random variables as above and using that $\Pr[\Xb = x] \leq 2^{-k}$ for all $x\in L$, we obtain
    \begin{align*}
        \Pr[\Zb \geq \beta] &= \Pr[\Zb \geq \E[\Zb] + \beta - \Pr[\Xb \in L]\cdot p] \\
        &\leq \Pr[\Zb \geq \E[\Zb] + \beta - p] \\
        &\leq 8\left(\frac{K p t + t^2}{K^2(\beta - p)^2}\right)^{t/2}.
    \end{align*}
    Now a union bound over the $2^{\binom{\ell}{\leq d}n}$ many degree-$d$ polynomial sources with $\ell$ inputs and $n$ outputs gives that the second condition for being an $(\alpha, \beta, k)$-isolator for such sources does not hold with probability at most $8\cdot 2^{\binom{\ell}{\leq d}n}\cdot \left(\frac{K p t + t^2}{K^2(\beta - p)^2}\right)^{t/2}$. Combining this with the failure probability for the first condition shows that $g$ fails to be an $(\alpha, \beta, k)$-isolator for such sources with probability at most 
    \[
        8\left(\left(\frac{Npt +t^2}{N^2(\alpha - p)^2 }\right)^{t/2} + 2^{\binom{\ell}{\leq d}n}\left(\frac{K p t + t^2}{K^2(\beta - p)^2}\right)^{t/2}\right) < 1
    \]
    by assumption. 
    Hence there exists some $\gb$ which is an $(\alpha, \beta, k)$-isolator, as desired.
\end{proof}

By iterating through a $t$-wise uniform family of hash functions, we must find an isolator in a reasonable amount of time.

\begin{corollary}\label{cor:f2_tind_rext}
    Let $k, n, d$, and $m$ be positive integers. Suppose $k \geq 10(d + \log n)$, $k \leq n$, and $m \leq 0.01k$. For $p = 2^{-m}$, there exist $\alpha = p - 2^{-\Omega(n)}$ and $\beta = p+2^{-\Omega(k)}$ such that there exists an $(\alpha, \beta, k)$-isolator $\iso$ for the class of degree-$d$ polynomials with $n$ outputs, which can be computed in time $2^{O\left(\binom{\Theta(k)}{\leq d}n^2\right)}$.
\end{corollary}
\begin{proof}
    Set $\ell = 5k$, $K = 2^k$, and $t=2\left(\binom{\ell}{\leq d}n + 4\right)$; note that $K^{0.99} \geq 4t$.
    By using \cref{lem:bdd_f2_tind_rext} under the conditions on $k, d$ in the statement, there exists a function $\iso$ which is an $(\alpha', \beta', k-1)$-isolator for the class of degree-$d$ polynomial sources with $4(k+1) \leq 5k$ inputs and $n$ outputs, where $\alpha' = p - 2^{-\Omega(n)}$ and $\beta' = p+2^{-\Omega(k)}$.
    Furthermore by \cref{lem:f2_input_red}, $\iso$ is an $(\alpha, \beta, k)$-isolator for the class of degree-$d$ polynomial sources with $n$ outputs, where $\alpha = \alpha'=p - 2^{-\Omega(n)}$ and $\beta = \beta' + 2^{-k} = p+2^{-\Omega(k)}$.

    We compute such a function by going over any fixed $t$-wise uniform family of hash functions from $\bin^n$ to $\bin^m$ until we find a function that lets us construct an isolator as described in  \cref{lem:bdd_f2_tind_rext}. 
    It is well known that there is such a family $\Hcal$ of size $2^{tn}$ where each function in the family can be evaluated in $\poly(n, m, t)$ time (see, for instance, \cite[ Corollary 3.34]{vadhan2012pseudorandomness}).
    For any such fixed function $h \in \Hcal$, we check whether $g$ defined by $g(x) = \indicator(h(x) = 0^m)$ is an $(\alpha, \beta, k)$-isolator. 
    This can be done in time $2^{O\left(\binom{\ell}{\leq d}n\right)}$.
\end{proof}

By invoking \cref{thm:main} with the above isolator, we obtain an explicit hard distribution for polynomial sources, proving \Cref{thm:hard_dist_poly}.
We restate the theorem below for the reader's convenience.

\thmharddistpoly*

\begin{proof}
We choose parameters to optimize the final distance with respect to the distribution length $(n+1)t$ while ensuring the construction takes only $\poly(N)$ time.
Set $n = \floorbra{(\log N)^{1/(2+\lambda)}}$ for a constant $\lambda > 0$ to be determined later. Let $d = \log n$ and $t = \floorbra{d/\Delta} - 1$ so that $\Delta(t+1) \leq d$. By picking $\delta$ to be small enough, we have $t \geq 1$. Let $k = \ceilbra{20\log n}$ so that the condition $k \geq 10(\log n + d)$ in \cref{cor:f2_tind_rext} is satisfied. Let $m$ be the largest integer such that for $p = 2^{-m}$, we have $t \geq \frac{1}{p} \log \frac{1}{p}$. Since $t \leq d = \log n$ and $m \leq \log t$, we have $m \leq 0.01k$. Now let $\iso$ be the $(\alpha, \beta, k)$-isolator given by \cref{cor:f2_tind_rext} with the chosen parameters and $\alpha = p - 2^{-\Omega(n)}, \beta = p+2^{-\Omega(k)}$.

Let $\Xb \in \Xcal$. Note that $\add_{n, t}(\Xb)$ can be computed by a degree $\Delta(t+1) \leq d$ polynomial source since $\add_{n, t}$ can be computed by a degree-$(t+1)$ polynomial. By \cref{thm:main}, we have
\begin{align*}
    1 - \tvdist{\Db - \Xb} &\leq (1-\alpha)^{t} + \beta + (2t^3+1)2^{-(n-k)} \\
    &\leq (1-p+2^{-\Omega(n)})^t + p+ 2^{-\Omega(k)} + (2t^3+1)\cdot 2^{-\Omega(n)} \\
    &\leq (1-p)^t + t \cdot2^{-\Omega(n)} + p+ 2^{-\Omega(k)} + 2^{-\Omega(n)} \\
    &\leq \exp(-pt) + p + n^{-\Omega(1)} \\
    &\leq O(p) \leq O\left(\frac{\log t}{t}\right) \\
    &\leq O\left(\frac{\Delta}{\log \log N} \log \left(\frac{\log \log N}{\Delta}\right)\right).
\end{align*}

Finally, we show that $\Db$ can be sampled in time $\poly(N)$. 
The isolator $\iso$ can be computed in time $2^{O\left(\binom{\Theta(k)}{\leq d}n^2\right)}$. 
We have
\[
    \binom{\Theta(k)}{\leq d} \leq \left(\frac{e \cdot \Theta(k)}{d}\right)^d \leq C_0^{\log n} = n^{\log C_0}
\]
for some constant $C_0 > 1$. 
Set $\lambda = \log C_0$, so $\iso$ can be computed in time $2^{O(n^{2+\lambda})} = \poly(N)$. 
Given that $\iso$ is computable in $\poly(N)$ time, it is straightforward to see that $\Db$ can be sampled in $\poly(N)$ time.

We now specialize to the case of $t = 1$ and $\tilde{\Db} = (\Ub^n, \tilde{\iso}(\Ub^n))$, where $\tilde{\iso}$ is an $(1/2 - 2^{-\Omega(n)}, 1/2+n^{-\Omega(1)}, 20\log n)$-isolator. Invoking \cref{thm:main}, we obtain
\[
    \tvdist{\tilde{\Db} - \Yb} \geq 1- (1/2 + 2^{-\Omega(n)})^2 - (1/2 + n^{-\Omega(1)})- 2^{-\Omega(n)} \geq 1/4 - n^{-\Omega(1)}. \qedhere
\]
\end{proof}

Note that our hard distribution above has at most $N$ bits instead of exactly $N$ bits. To obtain a distribution with exactly $N$ bits, one can pad with some bits that are fixed to, say, $0$. It is clear that the padded distribution does not become easier for polynomial sources of degree $\Delta$ since TV distance cannot increase by applying a projection. It is also easy to see that the padded distribution is computable in $\poly(N)$ time if the original distribution is computable in $\poly(N)$ time.

\subsection{Circuit and Local Sources}\label{ssec:circuit}

We will now instantiate \Cref{thm:main} in the case of sources derived from shallow circuits.
Below, it will be helpful to recall that $\nc$ circuits are (families of) constant-depth boolean circuits with bounded fan-in \textsf{AND, OR,} and \textsf{NOT} gates, while $\ac$ circuits allow the \textsf{AND} and \textsf{OR} gates to take an arbitrary number of inputs (and are otherwise defined the same as $\nc$ circuits).

\subsubsection{Circuit sources}

We begin with the more general \emph{circuit sources} before considering \emph{local sources}.

\begin{definition}[Circuit Source]\label{def:ac0_source}
     An $n$-bit source $\Xb$ is a \emph{circuit (or $\ac$) source} if there exists an $\ac$ circuit $C\colon \{0, 1\}^r \to \{0, 1\}^n$ such that $\Xb = C(\Ub^r)$.
     The \emph{size} and \emph{depth} of a circuit source $C(\Ub^r)$ are quantified by the number of gates and depth, respectively, of the circuit $C$.
\end{definition}

We highlight that hard distributions are known for such classes: the output of every small circuit\footnote{We will usually use ``small" to mean $\poly(n)$, at least in the context of circuits.} source has TV distance $1 - \exp(-\poly(n))$ from the uniform distribution over the codewords of any good code \cite{lovett2011bounded, beck2012large} and distance $\frac{1}{2} - \exp(-\poly(n))$ from the uniform distribution over input-output pairs of a particular extractor \cite{viola2014extractors, viola2020sampling}.
Our goal in this section is to recover similar bounds in a unified way.
We will ultimately obtain the following bounds.

\begin{restatable}{theorem}{thmharddistac}\label{thm:hard_dist_ac0}
    Let $\Delta$ be a positive integer. There exist $c, c', c_1 > 0$ depending on $\Delta$ such that the following holds. 
    Let $N$ and $S$ be positive integers, with  $N \leq S \leq \exp(N^{c})$. 
    There exist positive integers $n, t$, and $d$ satisfying $(n + 1)t \leq N$ and an isolator $\iso\colon \{0, 1\}^n \to \{0, 1\}$ for depth-$d$ circuit sources of size $\exp(n^{c'})$ on $\{0, 1\}^n$ with suitable parameters such that the following holds.

    Define the distribution
    \[
        \Db = (\Ub_1, \Ub_2, \dots \Ub_t, \iso(\Ub_1), \iso(\Ub_2), \dots \iso(\Ub_t)),
    \]
    where $\Ub_1, \Ub_2, \dots \Ub_t$ are independent copies of $\Ub^n$.
    Let $\Xcal$ be the class of all depth-$\Delta$, size-$S$ circuit sources.
    Then for any $\Xb \in \Xcal$,
    \[
        \tvdist{\Db - \Xb} \geq 1 - O\left(\frac{(\log S)^{{c_1}}}{N}\log\left(\frac{N}{(\log S)^{{c_1}}}\right)\right).
    \]   
    Moreover, the distribution $\Db$ can be sampled in $\poly(N)$ time.

    Additionally, there exists an isolator $\tilde{\iso} \colon \bin^n \to \bin$ (with possibly different parameters) such that the distribution $\tilde{\Db} = (\Ub^n, \tilde{\iso}(\Ub^n))$ satisfies $\tvdist{\tilde{\Db} - \Yb} \geq 1/4 - 2^{-n^{\Omega(1)}}$ for all depth-$\Delta$, size-$S$ circuit sources $\Yb$ on $\bin^{n+1}$.
    $\tilde{\Db}$, too, can be sampled in $\poly(N)$ time.
\end{restatable}

Recall our proof strategy for \Cref{thm:hard_dist_ac0} is to explicitly construct these isolators and apply \Cref{thm:main}.
We will need two existing tools for this.
The first is an explicit construction of an extractor for low-weight affine sources given by Rao \cite{rao2009extractors} and further optimized in \cite{viola2014extractors, de2012extractors}.
The work of \cite{de2012extractors, viola2014extractors} showed that such optimized versions also work as extractors for distributions generated by low-depth \emph{decision forests}, where we call a function $f\colon \bin^r \to \bin^n$ a depth-$d$ decision forest if each output bit is determined by a decision tree of depth at most $d$. 

\begin{theorem}[{\cite{rao2009extractors, de2012extractors, viola2014extractors}}] \label{thm:forest_ext}
    For a small constant $c > 0$, there exists an explicit extractor $\ext\colon\bin^n \rightarrow \bin^m$ for sources generated by depth-$(c \log n)$ decision forests with min-entropy $n^{0.9}$. Here, $m = n^{\Omega(1)}$ and the error is $2^{-m^{\Omega(1)}}$.
\end{theorem}

The second tool we require is Viola's min-entropy polarization result mentioned in \Cref{sec:main_pf}.

\begin{theorem}[\cite{viola2020sampling}]\label{thm:entropy_polarization}
    There exists a constant $c > 0$ such that the following holds. 
    If $C\colon \bin^r \rightarrow \bin^n$ is a depth-$d$ circuit of size at most $\exp(n^{c/d})$, then $C(\Ub^r)$ is $2^{-n^{\Omega(1)}}$-close to a distribution $\Yb$ such that
    \begin{itemize}
        \item $\Yb$ can be written as a mixture $\Yb = \sum_{i = 1}^\ell \frac{1}{\ell} \Yb_i$ of $\ell \leq 2^{n-n^{\Omega(1)}}$ distributions,
        \item For every $i \in [\ell]$, $\Yb_i$ is generated by a depth-$O(1)$ decision forest, and
        \item For every $i \in [\ell]$, $\Yb_i$ is constant or has min-entropy at least $n^{0.9}$.
    \end{itemize}    
\end{theorem}

We will use \Cref{thm:entropy_polarization} to show that the extractor from \Cref{thm:forest_ext} also works as a robust extractor, from which we can derive an isolator using \Cref{fct:iso_from_robust_ext}.

\begin{lemma} \label{lem:pol+ext->rob_ext2}
    Let $\Xcal, \Ycal$ be classes of distributions on $\bin^n$. 
    Suppose $\Ycal$ contains all constant distributions and the uniform distribution. 
    Additionally, suppose every source $\Xb \in \Xcal$ is $\gamma$-close to a distribution $\Yb$ such that
    \begin{itemize}
        \item $\Yb$ can be written as a mixture $\Yb = \sum_{i = 1}^\ell \frac{1}{\ell} \Yb_i$ where $\Yb_i \in \Ycal$ for all $i \in [\ell]$, and
        \item For all $i \in [\ell]$, $\Yb_i$ is either constant or has min-entropy at least $k$.
    \end{itemize}
    Let $k, k'$ be arbitrary positive reals.
    If $\ext$ is an $(\eps, k)$-extractor for $\Ycal$, then $\iso\colon\bin^n \to \bin$ defined by $\iso(x) = \indicator(\ext(x) = 0^m)$ is a $(2^{-m} - \eps, 2^{-m}+\eps+\ell \cdot 2^{-k'} + 2\gamma, k')$-isolator for $\Xcal$.
\end{lemma}

\begin{proof}
    Consider an arbitrary $\Xb \in \Xcal$, and define $L = \{x \in \bin^n: \Pr[\Xb = x] \leq 2^{-k'}\}$.
    We will verify the robustness condition required to use \cref{fct:iso_from_robust_ext} by bounding $\Pr[\Xb \in L \text{ and } \ext(\Xb) = 0^m]$.
    Let $\Yb = \sum_{i = 1}^\ell \frac{1}{\ell} \Yb_i$ be a mixture with the properties given by the assumption.
    Expanding the mixture and applying our distance assumption, we find that
    \begin{align*}
        \Pr[\Xb \in L \text{ and } \ext(\Xb) = 0^m] & \leq \Pr[\Yb \in L \text{ and } \ext(\Yb) = 0^m] + \gamma \\
        &\leq \sum_{i = 1}^\ell \frac{1}{\ell}\Pr[\Yb_i \in L \text{ and } \ext(\Yb_i) = 0^m] + \gamma.
    \end{align*}
    For each $\Yb_i$ with min-entropy $k$, we have $\Pr[\Yb_i \in L \text{ and } \ext(\Yb_i) = 0^m] \leq \Pr[\ext(\Yb_i) = 0^m] \leq 2^{-m} + \eps$ by the extractor guarantee.
    Additionally, for each $\Yb_i$ fixed to a constant outside $L$, we have $\Pr[\Yb_i \in L \text{ and } \ext(\Yb_i) = 0^m] = 0$.

    To analyze the remaining $\Yb_i$'s, let $S = \{x : x \in L \text{ and } \Yb_i \text{ is fixed to } x \text{ for some } i \in [\ell]\}$. Observe that $|S| \leq \ell$ since this is the number of parts in the mixture. 
    The total contribution to $\sum_{i} \frac{1}{\ell}\Pr[\Yb_i \in L \text{ and } \ext(\Yb_i) = 0^m]$ of all the constant points supported on $L$ is at most $\Pr[\Yb \in S]$. 
    Using our TV distance assumption, we can bound this quantity by $\Pr[\Yb \in S] \leq \Pr[\Xb \in S] + \gamma \leq |S|/2^{k'} + \gamma \leq \ell/2^{k'} + \gamma$, where the second inequality uses that $S$ only contains elements from $L$.

    Combining these three cases for $\Yb_i$, we obtain
    \[
        \Pr[\Xb \in L \text{ and } \ext(\Xb) = 0^m] \leq \frac{1}{2^m} + \eps + \frac{\ell}{2^{k'}}+2\gamma.
    \]
    Using \cref{fct:iso_from_robust_ext} finishes the proof.
\end{proof}

Instantiating \Cref{lem:pol+ext->rob_ext2} with the polarization property of \Cref{thm:entropy_polarization} and the extractor in \Cref{thm:forest_ext} produces the following corollary.

\begin{corollary}\label{cor:iso_ac0}
    There is a constant $c > 0$ such that the following holds. Let $d$ be a positive integer (treated as a constant). There exists $m_0 = n^{\Omega(1)}$ such that  for any positive integer $m \leq m_0$, there exists an explicit $(2^{-m} - 2^{-n^{\Omega(1)}}, 2^{-m} + 2^{-n^{\Omega(1)}}, n - n^{\Omega(1)})$-isolator for the class $\Xcal$ of depth-$d$ circuit sources of size at most $\exp(n^{c/d})$.
\end{corollary}
\begin{proof}
    Set $m_0 = n^{\Omega(1)}$ where the implicit constant is small enough. Let $\ext: \bin^n \rightarrow \bin^{m}$ be a $(2^{-n^{\Omega(1)}}, n^{0.9})$-extractor for the class $\Ycal$ of depth-$O(1)$ decision forests with output in $\bin^n$ given by \cref{thm:forest_ext}. 
    By \cref{thm:entropy_polarization} and \cref{lem:pol+ext->rob_ext2}, we get a $(2^{-m} - 2^{-n^{\Omega(1)}},2^{-m} - 2^{-n^{\Omega(1)}},n-n^{-\Omega(1)})$-isolator for $\Xcal$ from $\ext$.
\end{proof}

Now that we have a sufficiently good isolator, we can complete the proof of \Cref{thm:hard_dist_ac0}, restated below.

\thmharddistac*

\begin{proof}
We choose parameters to optimize the final distance with respect to the distribution length $(n+1)t$ while ensuring the construction takes only $\poly(N)$ time.
We will pick $c$, $c'$, and $c_1$ to be constants so that certain bounds below hold.
Let $n = \ceilbra{(\log S)^{{c_1}}}$. Let $t = \floorbra{N/(n+1)}$. By the assumption $S \leq \exp(N^{c})$, we have $t \geq 1$.
For any $\Xb$ of depth $\Delta$ and size at most $S$, $\add_{n, t}(\Xb)$ can be computed by a circuit of depth $d = \Delta+2$ and size $S + n(t+2) \leq \exp(n^{c'}) $ where we have used $S \geq N$.

Let $m$ be the largest integer such that for $p = 2^{-m}$, we have $t \geq \frac{1}{p} \log \frac{1}{p}$. 
Let $\iso\colon\bin^n \rightarrow \bin$ be a $(2^{-m} - 2^{-n^{\Omega(1)}}, 2^{-m} + 2^{-n^{\Omega(1)}}, n -n^{\Omega(1)})$-isolator for depth-$d$ circuit sources of size $\exp(n^{c'})$ on $\bin^n$ given by \cref{cor:iso_ac0}.
Now applying \cref{thm:main}, we obtain
\begin{align*}
1-\tvdist{\Xb - \Db} &\le (1-\alpha)^{t+1} +\beta +2^{-(n-k)}\pbra{2t^3 + 1} \\
&\le \exp(-p(t+1))+ p + 2^{-\Omega(n^{\Omega(1)})} \\
&\le O(p) \leq O\left(\frac{\log t}{t}\right) \\
&\le O\left(\frac{(\log S)^{c_1}}{N}\log\frac{N}{(\log S)^{c_1}}\right).
\end{align*}
In the second inequality, we  used that $t \leq 2^{ n^{c''}}$ for a small enough constant $c''$ depending on the constant in $k = n- n^{\Omega(1)}$.  
In the last inequality, we used $t =\Omega(N/(\log S)^{c_1})$.

We now specialize to the case of $t = 1$ and $\tilde{\Db} = (\Ub^n, \tilde{\iso}(\Ub^n))$, where $\tilde{\iso}$ is an $(1/2 - 2^{-n^{\Omega(1)}}, 1/2+2^{-n^{\Omega(1)}}, n - n^{\Omega(1)})$-isolator. Invoking \cref{thm:main}, we obtain
\[
    \tvdist{\tilde{\Db} - \Yb} \geq 1- (1/2 + 2^{-n^{\Omega(1)}})^2 - (1/2 + 2^{-n^{\Omega(1)}})- 2^{-n^{\Omega(1)}} \geq 1/4 - 2^{-n^{\Omega(1)}}. \qedhere
\]
\end{proof}

\subsubsection{Local sources}

We now consider \emph{local {\normalfont($\nc$)}} sources.

\begin{definition}[Local Source]\label{def:nc0_source}
    A function $f\colon\{0, 1\}^r \to \{0, 1\}^n$ is said to be \emph{$d$-local} if each output bit of $f$ depends on at most $d$ input bits. 
    An $n$-bit source $\textbf{X}$ is said to be $d$-local if there exists some $d$-local function $f\colon \{0, 1\}^r \to \{0, 1\}^n$ such that $\Xb = f(\Ub^r)$. 
\end{definition}

Observe that $O(1)$-local functions are precisely those expressible as $\nc$ circuits, so low locality sources are occasionally referred to as $\nc$ sources.
More generally, any $d$-local source over $\bin^n$ can be expressed as a depth-2 circuit source of size $n\cdot 2^d$, so \Cref{thm:hard_dist_ac0} implies the following result for such sources.
    
\begin{corollary}\label{cor:hard_dist_nc0}
    There exist $c, c', c_1 > 0$ such that the following holds. 
    Let $N$ and $\Delta$ be positive integers with $N\cdot 2^\Delta \leq \exp(N^{c})$.
    There exist positive integers $n, t$, and $d$ satisfying $(n + 1)t \leq N$ and an isolator, $\iso\colon \{0, 1\}^n \to \{0, 1\}$ for depth-$d$ circuit sources of size $\exp(n^{c'})$ on $\{0, 1\}^n$ with suitable parameters such that the following holds.

    Define the distribution
    \[
        \Db = (\Ub_1, \Ub_2, \dots \Ub_t, \iso(\Ub_1), \iso(\Ub_2), \dots \iso(\Ub_t)),
    \]
    where $\Ub_1, \Ub_2, \dots \Ub_t$ are independent copies of $\Ub^n$.
    Let $\Xcal$ be the class of all $\Delta$-local sources.
    Then for any $\Xb \in \Xcal$,
    \[
        \tvdist{\Db - \Xb} \geq 1 - O\left(\frac{(\log (N\cdot 2^{\Delta}))^{{c_1}}}{N}\log\left(\frac{N}{(\log (N\cdot 2^{\Delta}))^{{c_1}}}\right)\right).
    \]   
    Moreover, the distribution $\Db$ can be sampled in $\poly(N)$ time.

    Additionally, there exists an isolator $\tilde{\iso} \colon \bin^n \to \bin$ (with possibly different parameters) such that the distribution $\tilde{\Db} = (\Ub^n, \tilde{\iso}(\Ub^n))$ satisfies $\tvdist{\tilde{\Db} - \Yb} \geq 1/4 - 2^{-n^{\Omega(1)}}$ for all $\Delta$-local sources $\Yb$ on $\bin^{n+1}$.
    $\tilde{\Db}$, too, can be sampled in $\poly(N)$ time.
\end{corollary}

\subsection{Communication, Small-Space, and Turing-Machine Sources}\label{ssec:communication}

We now turn our attention to several related sources.

\subsubsection{Communication sources}

We begin by analyzing communication sources, since the two subsequent sources reduce to them.
For the unfamiliar reader, the basics of communication complexity can be found in the recent textbook \cite{RY20}.

\begin{definition}[Communication Source]\label{def:communication_source}
Alice and Bob hold private randomness $\Ab$ and $\Bb$, respectively. 
They exchange bits depending on their private randomness according to some protocol fixed beforehand. 
At the end, Alice outputs some random string $\Xb \in \bin^n$ depending on the transcript and $\Ab$. 
Similarly, Bob outputs $\Yb$ depending on the transcript and $\Bb$. 
Such a distribution $(\Xb, \Yb)$ is called a \emph{communication source}.
\end{definition}

The limitations of generating distributions with communication sources (and those that reduce to them) are fairly well understood \cite{AST+03, GW20, chattopadhyay2022space, yu2024sampling}.
It is known, for example, that any source sampled by a communication protocol using $\Omega(n)$ bits of communication has distance $1 - 2^{-\Omega(n)}$ from the uniform distribution over pairs of strings $(s,t)$ where $s \land t = 0^n$ \cite{GW20}, as well as distance $\frac{1}{2} - 2^{-\Omega(n)}$ from the uniform distribution over input-output pairs of the inner-product function \cite{chattopadhyay2022space}.\footnote{This latter work takes an approach somewhat similar to our own. While \cite{chattopadhyay2022space} use specific properties about communication sources (see the discussion in \Cref{sec:main_pf}) to directly reason about sampling, we use those properties to show that an extractor for such sources is also a robust extractor.}
The purpose of this section is to illustrate that these results can be (qualitatively) recovered in our unified framework.

\begin{restatable}{theorem}{thmharddistcommunication}\label{thm:hard_dist_communication}
    There exist constants $\delta, \delta_1, C > 0$ such that the following holds. Let $c$ be an integer satisfying $C\log N \le c \leq \delta_1 N$. 
    There exist parameters $n$, $t$ satisfying $nt \leq N$ and an isolator $\iso$ with suitable parameters for the class $\Ycal$ of communication sources on $(\bin^n)^2$ of cost $\delta n$ such that the following holds.

    Define the distribution 
    \[
        \Db=\big(\Ub_1^A, \Ub_2^A, \dots, \Ub_t^A, \Ub_1^B, \Ub_2^B, \dots, \Ub_t^B, \iso(\Ub_1^A, \Ub_1^B), \iso(\Ub_2^A, \Ub_2^B),\dots, \iso(\Ub_t^A, \Ub_t^B)\big),
    \]
    where $\Ub_1^A, \Ub_2^A, \dots, \Ub_t^A, \Ub_1^B, \Ub_2^B, \dots, \Ub_t^B$ are all independent copies of $\Ub^n$.
    Let $\Xcal$ be the class of all communication sources $(\Xb, \Yb)$ of cost $c$ where $\Xb$ is on $(\bin^n)^t$ and $\Yb$ is on $(\bin^n)^t \times \bin^t$. Then for any $(\Xb, \Yb) \in \Xcal$, we have
    \[
        \tvdist{\Db - (\Xb, \Yb)} \geq 1 - O\left(\frac{c}{N} \log \frac{N}{c}\right).
    \]
    Moreover, the distribution $\Db$ can be sampled in $\poly(N)$ time.

    Additionally, there exists an isolator $\tilde{\iso} \colon \bin^{2n} \to \bin$ (with possibly different parameters) such that the distribution $\tilde{\Db} = (\Ub^A, \Ub^B,\tilde{\iso}(\Ub^A,\Ub^B))$ satisfies $\tvdist{\tilde{\Db} - (\tilde{\Xb}, \tilde{\Yb})} \geq 1/4 - 2^{-\Omega(n)}$ for all communication sources $(\tilde{\Xb}, \tilde{\Yb})$ of cost $c$ where $\tilde{\Xb}$ is on $\bin^n$ and $\tilde{\Yb}$ is on $\bin^n \times \bin$.
    $\tilde{\Db}$, too, can be sampled in $\poly(N)$ time.
\end{restatable}

The following equivalent way of defining communication sources from \cite{AST+03} will be useful. 
A source $(\Xb, \Yb)$ generated by a communication protocol of cost $c$ is equivalent to a mixture of at most $2^c$ distributions $(\Xb_i, \Yb_i)$ where $\Xb_i$ and $\Yb_i$ are independent. 

As a first step toward constructing our isolator for \Cref{thm:hard_dist_communication}, we show that good two-source extractors are also robust extractors for sources $(\Xb, \Yb)$ where $\Xb$ and $\Yb$ are independent.
Recall that a two-source extractor $\ext$ is defined as in \Cref{def:ext_and_robust_ext}, only with the pseudorandom distribution being that $\tvdist{\ext(\Xb,\Yb) - \Ub^m} \le \eps$ for any two independent sources, each with some decent min-entropy. 

\begin{lemma}
    Let $\ext\colon \bin^n \times \bin^n \rightarrow \bin^m$ be a two-source extractor for min-entropy $k_0$ in each source and error $\eps$. Let $k,k_0, k_1, k_2$ satisfy $k > k_1 + k_2$ and $k_1 > k_0$. Consider any independent sources $\Xb, \Yb$ and define $L =  \{(x, y) \in (\bin^n)^2 : \Pr[(\Xb, \Yb) = (x, y)] \leq 2^{-k}\}$. Then for every $z \in \bin^m$, we have
    \begin{align*}
        \Pr[(\Xb, \Yb) \in L \text{ and } \ext(\Xb, \Yb) = z] &\leq \frac{1}{2^{k_2-n-1}}+\max\left(\frac{1}{2^m} + \eps, \frac{1}{2^{k_1-k_0}}\right).
    \end{align*}       
\end{lemma}
\begin{proof}
    Define the sets
    \begin{align*}
        L_{\Xb}^1 &= \{x \in \bin^n : 2^{-k_1} \geq \Pr[\Xb = x] > 2^{-k_2}\}, \\
        L_{\Xb}^2 &= \{x \in \bin^n : \Pr[\Xb = x] \leq 2^{-k_2}\},
    \end{align*}
    and define $L_{\Yb}^1$ and $L_{\Yb}^2$ similarly. 
    Using the assumption $k > k_1 + k_2$, observe that
    \[
        L \subseteq (L_{\Xb}^1 \times L_{\Yb}^1) \cup (L_{\Xb}^2 \times \bin^n) \cup (\bin^n \times L_{\Yb}^2).
    \]
    By the union bound, we have
    \begin{multline*}
        \Pr[(\Xb, \Yb) \in L \text{ and } \ext(\Xb, \Yb) = z] \\
        \leq \Pr[(\Xb, \Yb) \in L_{\Xb}^1 \times L_{\Yb}^1 \text{ and } \ext(\Xb, \Yb) = z] + \Pr[\Xb \in L_{\Xb}^2] + \Pr[\Yb \in L_{\Yb}^2].
    \end{multline*}
    Each of the last two terms in the sum can be trivially bounded by $2^n \cdot 2^{-k_2}$.

    For the remaining first term, we consider two cases. Define $p_1 = \Pr[\Xb \in L_{\Xb}^1]$ and $p_2 = \Pr[\Yb \in L_{\Yb}^1]$. 
    If both $p_1 \geq 2^{k_0 - k_1}$ and $p_2 \geq 2^{k_0 - k_1}$, then the independent sources $\tilde{\Xb}_1 \coloneqq \left(\Xb \mid \Xb \in L_{\Xb}^1\right)$ and $\tilde{\Xb}_2 \coloneqq \left(\Yb \mid \Yb \in L_{\Yb}^1\right)$ both have min-entropy at least $k_0$.
    In this case,
    \begin{align*}
        \Pr[(\Xb, \Yb) \in L_{\Xb}^1 \times L_{\Yb}^1 \text{ and } \ext(\Xb, \Yb) = z] &= \Pr[(\Xb, \Yb) \in L_{\Xb}^1 \times L_{\Yb}^1]\cdot \Pr[\ext(\tilde{\Xb}_1, \tilde{\Xb}_2) = z] \\
        &\leq p_1p_2\left(\frac{1}{2^m} + \eps\right) \leq \frac{1}{2^m} + \eps.
    \end{align*}
    On the other hand, if $\min\{p_1, p_2\} < 2^{k_0 - k_1}$, then 
    \[
    \Pr[(\Xb, \Yb) \in L_{\Xb}^1 \times L_{\Yb}^1 \text{ and } \ext(\Xb, \Yb) = z] \leq p_1 p_2 < \frac{1}{2^{k_1-k_0}}. 
    \]
    Combining these bounds gives the desired inequality.
\end{proof}

Combined with \cref{fct:iso_from_robust_ext}, the above lemma implies that a good enough two-source extractor can be used to construct an isolator for two independent sources. 
There is a mature line of work on two-source extractors (see \cite{li2023two} and references therein), but for our purposes, it suffices to use a universal hash function. 

The following connection is a well-known variant of the leftover hash lemma \cite{hastad1999pseudorandom}. A proof can be found, for instance, in \cite{lee2005extracting}. 
A simple example of a universal hash function is given by the inner product function over $\Fbb_{2^m}$ where we view both the inputs as vectors over $\Fbb_{2^m}$.
\begin{lemma}
    If $H \colon \bin^n \times \bin^n \rightarrow \bin^m$ is a universal hash function, then $H$ is also a two-source extractor for min-entropy $k$ in each source and error $\eps$ when $k \geq (n+m)/2 + \log(1/\eps)$.
\end{lemma}

This lemma and the earlier discussion imply the following.
\begin{corollary}\label{cor:isolator_for_two_sources}
Let $\Xcal$ denote the class of sources $(\Xb, \Yb)$ on $\bin^n \times \bin^n$ where $\Xb$ and $\Yb$ are independent. There is a constant $\delta > 0$ such that for any integer $m \leq \delta n$, there is an explicit $(2^{-m} - 2^{-\Omega(n)}, 2^{-m} + 2^{-\Omega(n)}, (1-\Omega(1))n)$-isolator for $\Xcal$.
\end{corollary}

The following lemma lets us obtain an isolator for communication sources from an isolator for two independent sources.

\begin{lemma}
    Let $\Xcal, \Ycal$ be classes of distributions on $\bin^n$. 
    Suppose every source $\Xb \in \Xcal$ can be written as a mixture $\Xb = \sum_{i = 1}^{2^\ell} p_i \Yb_i$ where $\Yb_i \in \Ycal$ for all $i \in [2^\ell]$.
    For any positive real $t$, if $\iso$ is an $(\alpha, \beta, k-t)$-isolator for $\Ycal$, then it is an $(\alpha,2^{-(t-\ell)}+\beta,k)$-isolator for $\Xcal$.
\end{lemma}
\begin{proof}
    By assumption, $\Pr[\iso(\Ub^n) = 1] \ge \alpha$, so it remains to verify the second condition of isolators.
    Let $L = \{x \in \bin^n : \Pr[\Xb = x] \leq 2^{-k}\}$ and $L_{\Yb_i} = \{y \in \bin^n : \Pr[\Yb_i = y] \leq 2^{-(k-t)}\}$ for all $i \in [2^\ell]$.
    Observe that $L \subseteq L_{\Yb_i}$ for every $i$ with $p_i \ge 2^{-t}$, since if $x \in L$, we have
    \[
        \Pr[\Yb_i = x] \le \frac{1}{p_i}\Pr[\Xb = x] \le 2^t \cdot 2^{-k} = 2^{-(k-t)}.
    \]
    Expanding out the mixture, we find
    \begin{align*}
        \Pr[\Xb \in L \text{ and } \iso(\Xb) = 1] &= \sum_{i = 1}^{2^\ell} p_i \Pr[\Yb_i \in L \text{ and } \iso(\Yb_i) = 1] \\
        &\le \sum_{i:p_i < 2^{-t}}p_i + \sum_{i:p_i \ge 2^{-t}} p_i\Pr[\Yb_i \in L_{\Yb_i} \text{ and } \iso(\Yb_i) = 1] \\
        &\le 2^{-(t-\ell)} + \beta. \qedhere
    \end{align*}
\end{proof}

Recalling that every communication source $(\Xb,\Yb)$ of cost $c$ can be expressed as a mixture of at most $2^c$ distributions over pairs of independent variables yields the following.

\begin{corollary}\label{cor:rext_communication}
There exists a constant $\delta > 0$ such that the following holds. Let $\Xcal$ be the class of all communication sources on $(\bin^n)^2$ of cost at most $\delta n$. For any integer $m \leq \delta n$, there is an explicit $(2^{-m} - 2^{-\Omega(n)}, 2^{-m} + 2^{-\Omega(n)}, (1-\Omega(1))n)$-isolator for $\Xcal$.
\end{corollary}

We can now prove our main result about communication sources, restated below.

\thmharddistcommunication*

\begin{proof}
We choose parameters to optimize the final distance with respect to the distribution length $(2n+1)t$ while ensuring the construction takes only $\poly(N)$ time.
Let $\delta$ be the constant in \cref{cor:rext_communication}. Let $n=\ceilbra{c/\delta_2 }$ for a small constant $\delta_2 > 0$ and $t = \floorbra{N/n}$. 
If $\delta_1 \leq \delta_2/2$, then $c \leq \delta_1 N$ implies $n \leq N$ and $t \geq 1$.
We have $\log(t+1) \leq \delta_3 n$ for a small constant $\delta_3 > 0$ because $n \geq c/\delta_2 \geq C\log N/\delta_2$ and $t \leq N-1$ as long as we pick the constants to satisfy $\delta_2/C \leq \delta_3$.
Let $m$ be the largest integer such that if we define $p = 2^{-m}$, then $t \geq \frac{1}{p} \log \frac{1}{p}$. Observe that $t = \Theta(m2^m)$. We have $ m\leq \log t \leq \delta_3 n \leq \delta n$ if we pick $\delta_3 \leq \delta$.
Let $\alpha = 2^{-m} - 2^{-\Omega(n)}, \beta= 2^{-m} + 2^{-\Omega(n)}, k = (1-\Omega(1))n$ be such that \cref{cor:rext_communication} gives an $(\alpha, \beta, k)$-isolator $\iso$ for $\Ycal$.

To apply \cref{thm:main}, we first verify that for any communication source $(\Xb, \Yb)$ in $\Xcal$, there is a communication source in $\Ycal$ generating $\add_{n, t}{(\Xb, \Yb)}$. Suppose we have a protocol $P$ of cost $c$ generating $(\Xb, \Yb)$, which we view as \[((\Xb_1, \Xb_2, \dots, \Xb_t), (\Yb_1, \dots, \Yb_t, \zb_1, \zb_2, \dots, \zb_t)).\]
The protocol $Q$ for $\add(\Xb, \Yb)$ works as follows. Start by running the protocol $P$ (but do not output according to it). At this point, Bob knows $(\Yb_1, \dots, \Yb_t, \zb_1, \zb_2, \dots, \zb_t)$ that he would output according to $P$. Let $i \in [t+1]$ be the smallest index such that $\zb_i = 1$. (If $\zb_j = 0$ for all $j$, then $i = t+1$.) Bob sends $i$ to Alice using $\ceilbra{\log(t+1)}$ bits. Now Alice outputs $\Xb_i$ where $(\Xb_1, \Xb_2, \dots, \Xb_t)$ is what we should have output based on the transcript according to protocol $P$. Similarly Bob outputs $\Yb_i$. The total cost is at most $c + \ceilbra{\log(t+1)} \leq \delta_2 n + \delta_3 n \leq \delta n$ for sufficiently small $\delta_2, \delta_3$.

Now applying \cref{thm:main}, we obtain
\begin{align*}
1-\tvdist{(\Xb, \Yb) - \Db} &\le (1-\alpha)^{t+1} +\beta +2^{-(n-k)}\pbra{2t^3 + 1} \\
&\le \exp(-p(t+1))+ p + 2^{-\Omega(n)} \\
&\le O(p) \leq O\left(\frac{\log t}{t}\right) \\
&\le O\left(\frac{c}{N}\log\frac{N}{c}\right).
\end{align*}
In the second inequality, we  used that $t \leq 2^{\delta_3 n}$ for small enough $\delta_3$ (depending on $k = (1-\Omega(1))n$).  
In the last inequality, we used $t =\Omega(N/c)$. 

We now specialize to the case of $t = 1$ and $\tilde{\Db} = (\Ub^A,\Ub^B, \tilde{\iso}(\Ub^A,\Ub^B))$ where $\tilde{\iso}$ is an $(1/2 - 2^{-\Omega(n)}, 1/2+2^{-\Omega(n)}, (1-\Omega(1))n)$-isolator. Invoking \cref{thm:main}, we obtain
\[
    \tvdist{\tilde{\Db} - (\tilde{\Xb}, \tilde{\Yb})} \geq 1/4 - 2^{-\Omega(n)}. \qedhere
\]
\end{proof}

\subsubsection{Small-space sources}

We now discuss small-space sources, which correspond to the streaming model of computation.
We will model these sources by read-once branching programs (ROBPs) as done by Chattopadhyay, Goodman, and Zuckerman \cite{chattopadhyay2022space}, although we note there is an alternative, albeit roughly equivalent \cite{chattopadhyay2022space}, model defined by Kamp, Rao, Vadhan, and Zuckerman \cite{KRVZ11}.

\begin{definition}[Small-Space Source]\label{def:small_space_source}
    A \emph{space-$s$ source} $\Xb$ on $\bin^n$ is a source generated by a width $2^s$ multi-output read-once branching program.
    More precisely, the branching program is viewed as a layered graph with $m$ layers (for some $m \geq 2$) with a single start vertex in the first layer and $2^s$ vertices in each subsequent layer. Each vertex has two outgoing edges, each labeled with some string in $\bin^*$. We assume that all strings in a layer have the same length.
    The source is generated by taking a random walk starting from the start vertex, picking the next edge uniformly at random at each step and outputting the corresponding strings on the edges in order.
\end{definition}

As observed in \cite{chattopadhyay2022space}, for any space-$s$ source $\Xb$ and any contiguous partition $\Xb = (\Xb_1, \Xb_2)$, there is a communication source of cost $s+1$ sampling $(\Xb_1, \Xb_2)$. 
We emphasize that while the length of $\Xb_1$, say $\ell$, is arbitrary, we require $\Xb_1$ to consist of the first $\ell$ bits of $\Xb$.
By combining this observation with \cref{thm:hard_dist_communication}, we immediately obtain the following.

\begin{corollary}\label{cor:hard_dist_small_space}
    There exist constants $\delta, \delta_1,  C > 0$ such that the following holds. 
    Let $s$ be an integer satisfying $C\log N \le s \leq \delta_1 N$. 
    There exist parameters $n$, $t$ satisfying $nt \leq N$ and an isolator $\iso$ with suitable parameters for the class of communication sources on $(\bin^n)^2$ of cost $\delta n$ such that the following holds.
    
    Define the distribution
    \[
        \Db=(\Ub_1^A, \Ub_2^A, \dots, \Ub_t^A, \Ub_1^B, \Ub_2^B, \dots, \Ub_t^B, \iso(\Ub_1^A, \Ub_1^B), \iso(\Ub_2^A, \Ub_2^B),\dots, \iso(\Ub_t^A, \Ub_t^B)),
    \]
    where $\Ub_1^A, \Ub_2^A, \dots, \Ub_t^A, \Ub_1^B, \Ub_2^B, \dots, \Ub_t^B$ are all independent copies of $\Ub^n$.
    Let $\Xcal$ be the class of all space-$s$ sources on $(\bin^n)^{2t} \times \bin^t$. Then for any $\Xb \in \Xcal$,
    \[
        \tvdist{\Db - \Xb} \geq 1 - O\left(\frac{s}{N} \log \frac{N}{s}\right).
    \]    
    Moreover, the distribution $\Db$ can be sampled in $\poly(N)$ time.
    
    Additionally, there exists an isolator $\tilde{\iso} \colon \bin^n \to \bin$ (with possibly different parameters) such that the distribution $\tilde{\Db} = (\Ub^n, \tilde{\iso}(\Ub^n))$ satisfies $\tvdist{\tilde{\Db} - \Yb} \geq 1/4 - 2^{-\Omega(n)}$ for all space-$s$ sources $\Yb$ on $\bin^{n+1}$.
    $\tilde{\Db}$, too, can be sampled in $\poly(N)$ time.
\end{corollary}

\subsubsection{Turing-machine sources}\label{ssec:TM}

The final sources we consider in this work are generated by Turing machines.

\begin{definition}[Turing-Machine Source]\label{def:TM_source}
We consider randomized Turing machines with one (read-write) tape and one tape head with alphabet $\bin$. 
The tape is initialized to all zeros.
In one step, based on a uniform random bit (independent of previous random bits), the current state, and the symbol at the current position, the machine writes to the cell, updates the state, and moves the head to an adjacent cell. 
We do not require the machine to halt.

A \emph{Turing-machine source} on $n$ bits given by a machine $M$ running in time $T$ is sampled as follows. Run $M$ for $T$ steps and output the first $n$ bits on its tape.
\end{definition}

Viola \cite{viola2012extractors} showed that Turing-machine sources can be simulated by communication sources.

\begin{lemma}[\cite{viola2012extractors}] \label{lem:tm_to_comm}
    Let $\Zb$ be an $n$-bit Turing-machine source sampled by a machine with $Q$ states in time $T$. 
    Let $n_A, n_B, b$ be positive integers satisfying $n_A + n_B = n$ and $b \leq n_B$. 
    Let $\Zb = (\Xb, \Yb)$ where $\Xb \in \bin^{n_A}$ and $\Yb \in \bin^{n_B}$. 
    There exists a communication protocol generating $(\Xb, \Yb)$ whose cost is $O(\log Q \cdot T \log T/b)+b$.
\end{lemma}
The above statement does not appear in exactly this form in \cite{viola2012extractors}, so we briefly explain how it is obtained from ideas there. Lemma 1.3 in \cite{viola2012extractors} shows that a Turing-machine source as in the above statement can be written as a convex combination of $2^{O(\log Q \cdot T \log T/b)}$ many pairs of independent sources. The statement there partitions the $n$ tape cells into regions of size $(\ell, b,\ell)$, but the proof also works if it is partitioned into $n_A, b, n_B - b$ where $n_A$ is not necessarily equal to $n_B - b$. For each pair $(\Xb_i, \Yb_i)$ occurring in this convex combination, $n_A \leq |\Xb_i| \leq n_A + b$. To obtain independent sources $(\Xb_i', \Yb_i')$ where $|\Xb_i'| = n_A, |\Yb_i'| = n_B$, we write $\Xb_i$ as a convex combination obtained by conditioning on the last $|\Xb_i| - n_A \leq b$ bits and after this conditioning, move these bits to the second source. This way of conditioning appears in the proof of Theorem 1.2 in \cite{viola2012extractors}.

Combining \cref{lem:tm_to_comm} with \cref{thm:hard_dist_communication}, we obtain the following.

\begin{corollary}\label{cor:hard_dist_TM}
    There exist constants $\delta, \delta_1, C > 0$ such that the following holds. Let $Q$ and $T$ be positive integers. Let $b = \ceilbra{\sqrt{\log Q \cdot T \log T}}$. Suppose $C\log N \le b \leq \delta_1 N$. There exist parameters $n$, $t$ satisfying $nt \leq N$ and an explicit isolator $\iso$ with suitable parameters for the class $\Ycal$ of communication sources on $(\bin^n)^2$ of cost $\delta n$ such that the following holds.
    
    Define the distribution
    \[
        \Db=(\Ub_1^A, \Ub_2^A, \dots, \Ub_t^A, \Ub_1^B, \Ub_2^B, \dots, \Ub_t^B, \iso(\Ub_1^A, \Ub_1^B), \iso(\Ub_2^A, \Ub_2^B),\dots, \iso(\Ub_t^A, \Ub_t^B)),
    \]
    where $\Ub_1^A, \Ub_2^A, \dots, \Ub_t^A, \Ub_1^B, \Ub_2^B, \dots, \Ub_t^B$ are all independent copies of $\Ub^n$.
    Let $\Xcal$ be the class of all Turing-machine sources $\Xb$ with $Q$ states and running in time $T$. Then for any $\Xb \in \Xcal$,
    \[
        \tvdist{\Db - \Xb} \geq 1 - O\left(\frac{\sqrt{\log Q \cdot T \log T}}{N} \log \frac{N}{\sqrt{\log Q \cdot T \log T}}\right).
    \]
    Moreover, the distribution $\Db$ can be sampled in $\poly(N)$ time.
    
    Additionally, there exists an isolator $\tilde{\iso} \colon \bin^n \to \bin$ (with possibly different parameters) such that the distribution $\tilde{\Db} = (\Ub^n, \tilde{\iso}(\Ub^n))$ satisfies $\tvdist{\tilde{\Db} - \Yb} \geq 1/4 - 2^{-\Omega(n)}$ for all Turing-machine sources $\Yb$ with $Q$ states and running in time $T$.
    $\tilde{\Db}$, too, can be sampled in $\poly(N)$ time.
\end{corollary}

\section{Open Problems}\label{sec:open_problems}

We conclude by highlighting several directions for future research.
One notable goal is to bring our understanding of sampling by small circuits with parity gates (i.e., $\ac[\oplus]$) up to the frontier of what is known in the computational setting.

\begin{enumerate}
    \item\label[question]{open_q1} Provide an explicit distribution over $\bin^n$ which has distance $1-o(1)$ from the output of any $\ac[\oplus]$ circuit with $n$ outputs, $\poly(n)$ many gates, and arbitrarily many random input bits.
\end{enumerate}

Even obtaining a much weaker bound would be novel; to the best of our knowledge, no hardness results for explicit\footnote{Existential results can be found in \Cref{app:non-constructive}.} distributions exist beyond the trivial arguments of finding a hard distribution for extremely small circuits via brute force, or appealing to the limitations posed by binary precision. 
For example, the $(1/3)$-biased distribution over $\bin$ cannot be exactly generated by a circuit with $r$ random bits as input, since each output occurs with probability that is an integer multiple of $2^{-r}$.
Of course, one could also ask for a similarly hard distribution in the case of circuits with mod $p$ gates for any other prime $p$.
Without these more powerful gates, it is known that such circuits cannot accurately generate the uniform distribution over the codewords of a good code \cite{lovett2011bounded, beck2012large}.
Unfortunately, the arguments do not seem to generalize.

In light of \Cref{thm:hard_dist_poly}, a reasonable approach to \Cref{open_q1} is to construct an isolator for $\Fbb_2$-polynomial sources of polylogarithmic degree, and then appeal to classical results on polynomial approximations of circuits \cite{razborov1987lower, smolensky1987algebraic}.
This is a more general degree constraint than the best-known explicit extractors can handle \cite{chattopadhyay2024extractors}, but we are optimistic further improvements in these extractors will be flexible enough for our techniques to apply.
It is also worth emphasizing that extractors and isolators are incomparable objects, and it is plausible that the latter may be easier to construct in certain settings.

An alternative strengthening of \Cref{thm:hard_dist_poly} is to improve the quantitative behavior; while we can prove a $1-o(1)$ bound, the specific decay rate of the $o(1)$ term is suboptimal.

\begin{enumerate}[resume]
    \item Provide an explicit distribution over $\bin^n$ which has distance $1 - \exp(-n^{\Omega_d(1)})$ from the output of any degree-$d$ $\Fbb_2$-polynomial source. \label{itm:Question2}
\end{enumerate}

We note that one can prove such a bound in the case of $d = 1$, where the source $P(\Ub^r)$ is uniform over an affine subspace.
Letting $\ext\colon\bin^n \to \bin^m$ be an $(\eps,k)$-extractor for affine sources, we can set the hard distribution $\Db$ to be uniform over the preimage $\ext^{-1}(0^m)$.
If $P(\Ub^r)$ has min-entropy at least $k$, then it has distance at least $1 - 2^{-m} - \eps$ from $\Db$.
Otherwise, $P(\Ub^r)$ is uniform over a set of size less than $2^k$, so it has distance at least $1 - 2^{-(n-k)}/(2^{-m}-\eps)$ from $\Db$.
By known explicit constructions of such objects \cite{Bou07, Yeh11, Li11}, we can take $k = \Omega(n)$, $\eps = 2^{-\Omega(n)}$, and $m=\Omega(n)$ with a sufficiently small implicit constant to conclude $\tvdist{P(\Ub^r) - \Db} \ge 1 - \exp(-\Omega(n))$.

Sadly, this argument already breaks down for $d=2$, since we no longer get meaningful information about $\supp{P(\Ub^r)}$ in the low min-entropy case.
Moreover, our current framework does not appear strong enough to address \hyperref[itm:Question2]{Question~\ref*{itm:Question2}}, and it would be interesting to determine whether the framework itself can be strengthened.

\begin{enumerate}[resume]
    \item Prove a quantitatively stronger version of \Cref{thm:main}.
\end{enumerate}

More specifically, it would be desirable to have the $\beta$ term in \Cref{thm:main} decrease as $t$ increases.
It also seems worthwhile to try to improve the $t=1$ setting to be able to obtain the optimal form $\frac{1}{2} - o(1)$, rather than our current $\frac{1}{4} - o(1)$.
In both cases, one may wish to consider a variant of extractors even beyond robust extractors or isolators.

\section{Acknowledgements}

We thank Jesse Goodman and Mohit Gurumukhani for a number of helpful comments on an earlier draft.
In particular, we are grateful for the suggestion of and discussion about ``smooth extractors'' as mentioned in \Cref{ssec:ext_and_rext}.
We also thank anonymous reviewers for their useful feedback.

\bibliographystyle{alpha}
\bibliography{ref}

@inproceedings{kane2024locality,
  title={Locality bounds for sampling {H}amming slices},
  author={Kane, Daniel M and Ostuni, Anthony and Wu, Kewen},
  booktitle={Proceedings of the 56th Annual ACM Symposium on Theory of Computing},
  pages={1279--1286},
  year={2024}
}

@article{viola2020sampling,
  title={Sampling lower bounds: boolean average-case and permutations},
  author={Viola, Emanuele},
  journal={SIAM Journal on Computing},
  volume={49},
  number={1},
  pages={119--137},
  year={2020},
  publisher={SIAM}
}

@inproceedings{chattopadhyay2022space,
  title={The space complexity of sampling},
  author={Chattopadhyay, Eshan and Goodman, Jesse and Zuckerman, David},
  booktitle={13th Innovations in Theoretical Computer Science Conference (ITCS)},
  year={2022}
}

@inproceedings{shaltiel2024explicit,
  title={Explicit codes for poly-size circuits and functions that are hard to sample on low entropy distributions},
  author={Shaltiel, Ronen and Silbak, Jad},
  booktitle={Proceedings of the 56th Annual ACM Symposium on Theory of Computing},
  pages={2028--2038},
  year={2024}
}

@incollection {WP26,
    AUTHOR = {Bene Watts, Adam and Parham, Natalie},
     TITLE = {Unconditional quantum advantage for sampling with shallow
              circuits},
 BOOKTITLE = {17th {I}nnovations in {T}heoretical {C}omputer {S}cience
              {C}onference},
    SERIES = {LIPIcs. Leibniz Int. Proc. Inform.},
    VOLUME = {362},
     PAGES = {Art. No. 17, 12},
 PUBLISHER = {Schloss Dagstuhl. Leibniz-Zent. Inform., Wadern},
      YEAR = {2026},
      ISBN = {978-3-95977-410-9},
   MRCLASS = {68Q12 (68Q06)},
  MRNUMBER = {5022595},
       DOI = {10.4230/lipics.itcs.2026.17},
       URL = {https://doi.org/10.4230/lipics.itcs.2026.17},
}

@article{viola2014extractors,
  title={Extractors for circuit sources},
  author={Viola, Emanuele},
  journal={SIAM Journal on Computing},
  volume={43},
  number={2},
  pages={655--672},
  year={2014},
  publisher={SIAM}
}

@inproceedings{lovett2011bounded,
  title={Bounded-depth circuits cannot sample good codes},
  author={Lovett, Shachar and Viola, Emanuele},
  booktitle={2011 IEEE 26th Annual Conference on Computational Complexity},
  pages={243--251},
  year={2011},
  organization={IEEE}
}

@inproceedings{beck2012large,
  title={Large deviation bounds for decision trees and sampling lower bounds for {AC0}-circuits},
  author={Beck, Chris and Impagliazzo, Russell and Lovett, Shachar},
  booktitle={2012 IEEE 53rd Annual Symposium on Foundations of Computer Science},
  pages={101--110},
  year={2012},
  organization={IEEE}
}

@phdthesis{haastad1986computational,
  title={Computational limitations for small depth circuits},
  author={H{\aa}stad, Johan},
  year={1986},
  school={Massachusetts Institute of Technology}
}

@inproceedings{smolensky1987algebraic,
  title={Algebraic methods in the theory of lower bounds for {B}oolean circuit complexity},
  author={Smolensky, Roman},
  booktitle={Proceedings of the nineteenth annual ACM symposium on Theory of computing},
  pages={77--82},
  year={1987}
}

@article{babai1987random,
  title={Random oracles separate {PSPACE} from the polynomial-time hierarchy},
  author={Babai, L{\'a}szi{\'o}},
  journal={Information Processing Letters},
  volume={26},
  number={1},
  pages={51--53},
  year={1987},
  publisher={Elsevier}
}

@article{boppana1987one,
  title={One-way functions and circuit complexity},
  author={Boppana, Ravi B and Lagarias, Jeffrey C},
  journal={Information and Computation},
  volume={74},
  number={3},
  pages={226--240},
  year={1987},
  publisher={Elsevier}
}

@article{viola2012complexity,
  title={The complexity of distributions},
  author={Viola, Emanuele},
  journal={SIAM Journal on Computing},
  volume={41},
  number={1},
  pages={191--218},
  year={2012},
  publisher={SIAM}
}

@inproceedings{viola2023new,
  title={New sampling lower bounds via the separator},
  author={Viola, Emanuele},
  booktitle={38th Computational Complexity Conference (CCC 2023)},
  year={2023},
  organization={Schloss Dagstuhl-Leibniz-Zentrum f{\"u}r Informatik},
}

@inproceedings{yu2024sampling,
  title={Sampling, flowers and communication},
  author={Yu, Huacheng and Zhan, Wei},
  booktitle={15th Innovations in Theoretical Computer Science Conference (ITCS 2024)},
  pages={100--1},
  year={2024},
  organization={Schloss Dagstuhl--Leibniz-Zentrum f{\"u}r Informatik}
}

@article{de2012extractors,
  title={Extractors and lower bounds for locally samplable sources},
  author={De, Anindya and Watson, Thomas},
  journal={ACM Transactions on Computation Theory (TOCT)},
  volume={4},
  number={1},
  pages={1--21},
  year={2012},
  publisher={ACM New York, NY, USA}
}

@inproceedings{viola2012extractors,
  title={Extractors for {T}uring-machine sources},
  author={Viola, Emanuele},
  booktitle={International Workshop on Approximation Algorithms for Combinatorial Optimization},
  pages={663--671},
  year={2012},
  organization={Springer}
}

@article{viola2016quadratic,
  title={Quadratic maps are hard to sample},
  author={Viola, Emanuele},
  journal={ACM Transactions on Computation Theory (TOCT)},
  volume={8},
  number={4},
  pages={1--4},
  year={2016},
  publisher={ACM New York, NY, USA}
}

@inproceedings{chattopadhyay2024extractors,
  title={Extractors for polynomial sources over {$\mathbb{F}_2$}},
  author={Chattopadhyay, Eshan and Goodman, Jesse and Gurumukhani, Mohit},
  booktitle={15th Innovations in Theoretical Computer Science Conference (ITCS)},
  pages={28--1},
  year={2024},
  organization={Schloss Dagstuhl--Leibniz-Zentrum f{\"u}r Informatik}
}

@incollection {GGH+24,
    AUTHOR = {Golovnev, Alexander and Guo, Zeyu and Hatami, Pooya and
              Nagargoje, Satyajeet and Yan, Chao},
     TITLE = {Hilbert functions and low-degree randomness extractors},
 BOOKTITLE = {Approximation, randomization, and combinatorial optimization.
              {A}lgorithms and techniques},
    SERIES = {LIPIcs. Leibniz Int. Proc. Inform.},
    VOLUME = {317},
     PAGES = {Art. No. 41, 24},
 PUBLISHER = {Schloss Dagstuhl. Leibniz-Zent. Inform., Wadern},
      YEAR = {2024},
      ISBN = {978-3-95977-348-5},
   MRCLASS = {68Q06},
  MRNUMBER = {4802551},
       DOI = {10.4230/lipics.approx/random.2024.41},
       URL = {https://doi.org/10.4230/lipics.approx/random.2024.41},
}

@incollection {AGMR25,
    AUTHOR = {Alrabiah, Omar and Goodman, Jesse and Mosheiff, Jonathan and
              Ribeiro, Jo\~{a}o},
     TITLE = {Low-degree polynomials are good extractors},
 BOOKTITLE = {Approximation, randomization, and combinatorial optimization.
              {A}lgorithms and techniques},
    SERIES = {LIPIcs. Leibniz Int. Proc. Inform.},
    VOLUME = {353},
     PAGES = {Art. No. 38, 25},
 PUBLISHER = {Schloss Dagstuhl. Leibniz-Zent. Inform., Wadern},
      YEAR = {2025},
      ISBN = {978-3-95977-397-3},
   MRCLASS = {68Q87},
  MRNUMBER = {4967596},
       DOI = {10.4230/lipics.approx/random.2025.38},
       URL = {https://doi.org/10.4230/lipics.approx/random.2025.38},
}

@inproceedings{trevisan2000extracting,
  title={Extracting randomness from samplable distributions},
  author={Trevisan, Luca and Vadhan, Salil},
  booktitle={Proceedings 41st Annual Symposium on Foundations of Computer Science},
  pages={32--42},
  year={2000},
  organization={IEEE}
}

@article {KRVZ11,
    AUTHOR = {Kamp, Jesse and Rao, Anup and Vadhan, Salil and Zuckerman,
              David},
     TITLE = {Deterministic extractors for small-space sources},
   JOURNAL = {J. Comput. System Sci.},
  FJOURNAL = {Journal of Computer and System Sciences},
    VOLUME = {77},
      YEAR = {2011},
    NUMBER = {1},
     PAGES = {191--220},
      ISSN = {0022-0000,1090-2724},
   MRCLASS = {68Q87 (68W20)},
  MRNUMBER = {2767133},
       DOI = {10.1016/j.jcss.2010.06.014},
       URL = {https://doi.org/10.1016/j.jcss.2010.06.014},
}

@inproceedings{ball2025extractors,
  title={Extractors for Samplable Distributions with Low Min-Entropy},
  author={Ball, Marshall and Shaltiel, Ronen and Silbak, Jad},
  booktitle={Proceedings of the 57th Annual ACM Symposium on Theory of Computing},
  pages={596--603},
  year={2025}
}

@inproceedings{bellare1994randomness,
  title={Randomness-efficient oblivious sampling},
  author={Bellare, Mihir and Rompel, John},
  booktitle={Proceedings 35th Annual Symposium on Foundations of Computer Science},
  pages={276--287},
  year={1994},
  organization={IEEE}
}

@article {FSS84,
    AUTHOR = {Furst, Merrick and Saxe, James B. and Sipser, Michael},
     TITLE = {Parity, circuits, and the polynomial-time hierarchy},
   JOURNAL = {Math. Systems Theory},
  FJOURNAL = {Mathematical Systems Theory. An International Journal on
              Mathematical Computing Theory},
    VOLUME = {17},
      YEAR = {1984},
    NUMBER = {1},
     PAGES = {13--27},
      ISSN = {0025-5661},
   MRCLASS = {68Q25 (94C15)},
  MRNUMBER = {738749},
MRREVIEWER = {Claus-Peter\ Schnorr},
       DOI = {10.1007/BF01744431},
       URL = {https://doi.org/10.1007/BF01744431},
}

@article {Ajt83,
    AUTHOR = {Ajtai, Mikl{\'{o}}s},
     TITLE = {{$\Sigma \sp{1}\sb{1}$}-formulae on finite structures},
   JOURNAL = {Ann. Pure Appl. Logic},
  FJOURNAL = {Annals of Pure and Applied Logic},
    VOLUME = {24},
      YEAR = {1983},
    NUMBER = {1},
     PAGES = {1--48},
      ISSN = {0168-0072,1873-2461},
   MRCLASS = {03C13 (03B10 03B15 03C35)},
  MRNUMBER = {706289},
MRREVIEWER = {A.\ H.\ Lachlan},
       DOI = {10.1016/0168-0072(83)90038-6},
       URL = {https://doi.org/10.1016/0168-0072(83)90038-6},
}

@inproceedings{yao1985separating,
  title={Separating the polynomial-time hierarchy by oracles},
  author={Yao, Andrew Chi-Chih},
  booktitle={26th Annual Symposium on Foundations of Computer Science (sfcs 1985)},
  pages={1--10},
  year={1985},
  organization={IEEE}
}

@inproceedings{hastad1986almost,
  title={Almost optimal lower bounds for small depth circuits},
  author={H{\aa}stad, Johan},
  booktitle={Proceedings of the eighteenth annual ACM symposium on Theory of computing},
  pages={6--20},
  year={1986}
}

@article{razborov1987lower,
  title={Lower bounds on the size of bounded depth circuits over a complete basis with logical addition},
  author={Razborov, Alexander A},
  journal={Mathematical Notes of the Academy of Sciences of the USSR},
  volume={41},
  number={4},
  pages={333--338},
  year={1987},
  publisher={Springer}
}

@article{HH24,
    author = {Hatami, Pooya and Hoza, William},
    title = {Theory of Unconditional Pseudorandom Generators},
    journal = {Foundations and Trends in Theoretical Computer Science},
    volume = {16},
    number = {1-2},
    pages = {1-210},
    year = {2024},
    month = {02},
    issn = {1551-305X},
    doi = {10.1561/0400000109},
    url = {https://www.emerald.com/fttcs/article-pdf/16/1-2/1/11150456/0400000109en.pdf}
}

@inproceedings {ADOY25,
    AUTHOR = {Anshu, Anurag and Dong, Yangjing and Ou, Fengning and Yao,
              Penghui},
     TITLE = {On the computational power of {QAC}0 with barely superlinear
              ancillae},
 BOOKTITLE = {S{TOC}'25---{P}roceedings of the 57th {A}nnual {ACM}
              {S}ymposium on {T}heory of {C}omputing},
     PAGES = {1476--1487},
 PUBLISHER = {ACM, New York},
      YEAR = {2025},
      ISBN = {979-8-4007-1510-5},
   MRCLASS = {68Q06},
  MRNUMBER = {4928535},
       DOI = {10.1145/3717823.3718189},
       URL = {https://doi.org/10.1145/3717823.3718189},
}

@inproceedings {NPVY24,
    AUTHOR = {Nadimpalli, Shivam and Parham, Natalie and Vasconcelos,
              Francisca and Yuen, Henry},
     TITLE = {On the {P}auli spectrum of {QAC}0},
 BOOKTITLE = {S{TOC}'24---{P}roceedings of the 56th {A}nnual {ACM}
              {S}ymposium on {T}heory of {C}omputing},
     PAGES = {1498--1506},
 PUBLISHER = {ACM, New York},
      YEAR = {2024},
      ISBN = {979-8-4007-0383-6},
   MRCLASS = {68Q12},
  MRNUMBER = {4764925},
       DOI = {10.1145/3618260.3649662},
       URL = {https://doi.org/10.1145/3618260.3649662},
}

@article{JTVW25,
  title={Improved Lower Bounds for {QAC$^0$}},
  author={Joshi, Malvika Raj and Tal, Avishay and Vasconcelos, Francisca and Wright, John},
  journal={arXiv preprint arXiv:2512.14643},
  year={2025}
}

@article{FGPT25,
  title={Tight bounds on depth-2 {QAC}-circuits computing parity},
  author={Fenner, Stephen and Grier, Daniel and Padé, Daniel and Thierauf, Thomas},
  journal={arXiv preprint arXiv:2504.06433},
  year={2025}
}

@book {RY20,
    AUTHOR = {Rao, Anup and Yehudayoff, Amir},
     TITLE = {Communication complexity and applications},
 PUBLISHER = {Cambridge University Press, Cambridge},
      YEAR = {2020},
     PAGES = {xvii+251},
      ISBN = {978-1-108-49798-5},
   MRCLASS = {68-02 (03F20 68Q06 68Q11 94A05)},
  MRNUMBER = {4312803},
}

@article{GMW26,
      title={$\mathsf{QAC}^0$ Contains $\mathsf{TC}^0$ (with Many Copies of the Input)}, 
      author={Daniel Grier and Jackson Morris and Kewen Wu},
      journal={arXiv preprint arXiv:2601.03243},
      year={2026}
}

@article {GW20,
    AUTHOR = {G\"{o}\"{o}s, Mika and Watson, Thomas},
     TITLE = {A lower bound for sampling disjoint sets},
   JOURNAL = {ACM Trans. Comput. Theory},
  FJOURNAL = {ACM Transactions on Computation Theory},
    VOLUME = {12},
      YEAR = {2020},
    NUMBER = {3},
     PAGES = {Art. 20, 13},
      ISSN = {1942-3454,1942-3462},
   MRCLASS = {68Q17 (68Q11)},
  MRNUMBER = {4144856},
MRREVIEWER = {Penghui\ Yao},
       DOI = {10.1145/3404858},
       URL = {https://doi.org/10.1145/3404858},
}

@article {AST+03,
    AUTHOR = {Ambainis, Andris and Schulman, Leonard J. and Ta-Shma, Amnon
              and Vazirani, Umesh and Wigderson, Avi},
     TITLE = {The quantum communication complexity of sampling},
   JOURNAL = {SIAM J. Comput.},
  FJOURNAL = {SIAM Journal on Computing},
    VOLUME = {32},
      YEAR = {2003},
    NUMBER = {6},
     PAGES = {1570--1585},
      ISSN = {0097-5397,1095-7111},
   MRCLASS = {94A15 (68Q05 68Q10 68Q15 81P68 94A20)},
  MRNUMBER = {2034251},
MRREVIEWER = {M.\ S.\ Burgin},
       DOI = {10.1137/S009753979935476},
       URL = {https://doi.org/10.1137/S009753979935476},
}

@article{KOW25,
      title={Symmetric Distributions from Shallow Circuits}, 
      author={Daniel M. Kane and Anthony Ostuni and Kewen Wu},
      year={2025},
      journal={arXiv preprint arXiv:2511.14127}
}

@inproceedings{alekseev2025sampling,
  title={Sampling Permutations with Cell Probes is Hard},
  author={Yaroslav Alekseev and Mika G\"o\"os and Konstantin Myasnikov and Artur Riazanov and Dmitry Sokolov},
  booktitle={Proceedings of the 58th Annual ACM Symposium on Theory of Computing},
  year={2026}
}

@incollection {GKM+26,
    AUTHOR = {Grier, Daniel and Kane, Daniel M. and Morris, Jackson and
              Ostuni, Anthony and Wu, Kewen},
     TITLE = {Quantum advantage from sampling shallow circuits: beyond
              hardness of marginals},
 BOOKTITLE = {17th {I}nnovations in {T}heoretical {C}omputer {S}cience
              {C}onference},
    SERIES = {LIPIcs. Leibniz Int. Proc. Inform.},
    VOLUME = {362},
     PAGES = {Art. No. 73, 14},
 PUBLISHER = {Schloss Dagstuhl. Leibniz-Zent. Inform., Wadern},
      YEAR = {2026},
      ISBN = {978-3-95977-410-9},
   MRCLASS = {68Q06 (68Q12)},
  MRNUMBER = {5022651},
       DOI = {10.4230/lipics.itcs.2026.73},
       URL = {https://doi.org/10.4230/lipics.itcs.2026.73},
}

@incollection {Vin25,
    AUTHOR = {Kumar, Vinayak M.},
     TITLE = {New pseudorandom generators and correlation bounds using
              extractors},
 BOOKTITLE = {16th {I}nnovations in {T}heoretical {C}omputer {S}cience
              {C}onference},
    SERIES = {LIPIcs. Leibniz Int. Proc. Inform.},
    VOLUME = {325},
     PAGES = {Art. No. 68, 23},
 PUBLISHER = {Schloss Dagstuhl. Leibniz-Zent. Inform., Wadern},
      YEAR = {2025},
      ISBN = {978-3-95977-361-4},
   MRCLASS = {68Q06},
  MRNUMBER = {4868470},
       DOI = {10.4230/lipics.itcs.2025.68},
       URL = {https://doi.org/10.4230/lipics.itcs.2025.68},
}

@incollection {HL25,
    AUTHOR = {Hoza, William M. and Lv, Zelin},
     TITLE = {On sums of {INW} pseudorandom generators},
 BOOKTITLE = {Approximation, randomization, and combinatorial optimization.
              {A}lgorithms and techniques},
    SERIES = {LIPIcs. Leibniz Int. Proc. Inform.},
    VOLUME = {353},
     PAGES = {Art. No. 67, 24},
 PUBLISHER = {Schloss Dagstuhl. Leibniz-Zent. Inform., Wadern},
      YEAR = {2025},
      ISBN = {978-3-95977-397-3},
   MRCLASS = {68Q87},
  MRNUMBER = {4967625},
       DOI = {10.4230/lipics.approx/random.2025.67},
       URL = {https://doi.org/10.4230/lipics.approx/random.2025.67},
}

@incollection {DH25,
    AUTHOR = {Doron, Dean and Hoza, William M.},
     TITLE = {Implications of better {PRG}s for permutation branching
              programs},
 BOOKTITLE = {Approximation, randomization, and combinatorial optimization.
              {A}lgorithms and techniques},
    SERIES = {LIPIcs. Leibniz Int. Proc. Inform.},
    VOLUME = {353},
     PAGES = {Art. No. 28, 20},
 PUBLISHER = {Schloss Dagstuhl. Leibniz-Zent. Inform., Wadern},
      YEAR = {2025},
      ISBN = {978-3-95977-397-3},
   MRCLASS = {68Q87},
  MRNUMBER = {4967586},
       DOI = {10.4230/lipics.approx/random.2025.28},
       URL = {https://doi.org/10.4230/lipics.approx/random.2025.28},
}

@incollection {LV25,
    AUTHOR = {Lee, Chin Ho and Viola, Emanuele},
     TITLE = {Pseudorandom bits for non-commutative programs},
 BOOKTITLE = {40th {C}omputational {C}omplexity {C}onference},
    SERIES = {LIPIcs. Leibniz Int. Proc. Inform.},
    VOLUME = {339},
     PAGES = {Art. No. 9, 22},
 PUBLISHER = {Schloss Dagstuhl. Leibniz-Zent. Inform., Wadern},
      YEAR = {2025},
      ISBN = {978-3-95977-379-9},
   MRCLASS = {68Q06 (20F05)},
  MRNUMBER = {4938509},
       DOI = {10.4230/lipics.ccc.2025.9},
       URL = {https://doi.org/10.4230/lipics.ccc.2025.9},
}

@incollection {DILV24,
    AUTHOR = {Derksen, Harm and Ivanov, Peter and Lee, Chin Ho and Viola,
              Emanuele},
     TITLE = {Pseudorandomness, symmetry, smoothing: {I}},
 BOOKTITLE = {39th {C}omputational {C}omplexity {C}onference},
    SERIES = {LIPIcs. Leibniz Int. Proc. Inform.},
    VOLUME = {300},
     PAGES = {Art. No. 18, 27},
 PUBLISHER = {Schloss Dagstuhl. Leibniz-Zent. Inform., Wadern},
      YEAR = {2024},
      ISBN = {978-3-95977-331-7},
   MRCLASS = {68Q87},
  MRNUMBER = {4774697},
       DOI = {10.4230/lipics.ccc.2024.18},
       URL = {https://doi.org/10.4230/lipics.ccc.2024.18},
}

@incollection {Sha25,
    AUTHOR = {Shaltiel, Ronen},
     TITLE = {Multiplicative extractors for samplable distributions},
 BOOKTITLE = {40th {C}omputational {C}omplexity {C}onference},
    SERIES = {LIPIcs. Leibniz Int. Proc. Inform.},
    VOLUME = {339},
     PAGES = {Art. No. 22, 22},
 PUBLISHER = {Schloss Dagstuhl. Leibniz-Zent. Inform., Wadern},
      YEAR = {2025},
      ISBN = {978-3-95977-379-9},
   MRCLASS = {68Q87},
  MRNUMBER = {4938522},
       DOI = {10.4230/lipics.ccc.2025.22},
       URL = {https://doi.org/10.4230/lipics.ccc.2025.22},
}

@inproceedings{rao2009extractors,
  title={Extractors for low-weight affine sources},
  author={Rao, Anup},
  booktitle={2009 24th Annual IEEE Conference on Computational Complexity},
  pages={95--101},
  year={2009},
  organization={IEEE}
}

@article{hastad1999pseudorandom,
  title={A pseudorandom generator from any one-way function},
  author={H{\aa}stad, Johan and Impagliazzo, Russell and Levin, Leonid A and Luby, Michael},
  journal={SIAM Journal on Computing},
  volume={28},
  number={4},
  pages={1364--1396},
  year={1999},
  publisher={SIAM}
}

@inproceedings{li2023two,
  title={Two source extractors for asymptotically optimal entropy, and (many) more},
  author={Li, Xin},
  booktitle={2023 IEEE 64th Annual Symposium on Foundations of Computer Science (FOCS)},
  pages={1271--1281},
  year={2023},
  organization={IEEE}
}

@article{vadhan2012pseudorandomness,
  title={Pseudorandomness},
  author={Vadhan, Salil P},
  journal={Foundations and Trends{\textregistered} in Theoretical Computer Science},
  volume={7},
  number={1-3},
  pages={1--336},
  year={2012},
  publisher={Emerald Publishing Limited}
}

@article {Bou07,
    AUTHOR = {Bourgain, Jean},
     TITLE = {On the construction of affine extractors},
   JOURNAL = {Geom. Funct. Anal.},
  FJOURNAL = {Geometric and Functional Analysis},
    VOLUME = {17},
      YEAR = {2007},
    NUMBER = {1},
     PAGES = {33--57},
      ISSN = {1016-443X,1420-8970},
   MRCLASS = {68R10 (11L07 60C05 68W20)},
  MRNUMBER = {2306652},
MRREVIEWER = {Alessandro\ Conflitti},
       DOI = {10.1007/s00039-007-0593-z},
       URL = {https://doi.org/10.1007/s00039-007-0593-z},
}

@article {Yeh11,
    AUTHOR = {Yehudayoff, Amir},
     TITLE = {Affine extractors over prime fields},
   JOURNAL = {Combinatorica},
  FJOURNAL = {Combinatorica. An International Journal on Combinatorics and
              the Theory of Computing},
    VOLUME = {31},
      YEAR = {2011},
    NUMBER = {2},
     PAGES = {245--256},
      ISSN = {0209-9683,1439-6912},
   MRCLASS = {11T23 (11L07)},
  MRNUMBER = {2848253},
MRREVIEWER = {Alessandro\ Conflitti},
       DOI = {10.1007/s00493-011-2604-9},
       URL = {https://doi.org/10.1007/s00493-011-2604-9},
}

@incollection {Li11,
    AUTHOR = {Li, Xin},
     TITLE = {A new approach to affine extractors and dispersers},
 BOOKTITLE = {26th {A}nnual {IEEE} {C}onference on {C}omputational
              {C}omplexity},
     PAGES = {137--147},
 PUBLISHER = {IEEE Computer Soc., Los Alamitos, CA},
      YEAR = {2011},
      ISBN = {978-0-7695-4411-3},
   MRCLASS = {68Q87 (68W20)},
  MRNUMBER = {3025368},
}

@article{lee2005extracting,
  title={Extracting randomness from multiple independent sources},
  author={Lee, Chia-Jung and Lu, Chi-Jen and Tsai, Shi-Chun and Tzeng, Wen-Guey},
  journal={IEEE Transactions on Information Theory},
  volume={51},
  number={6},
  pages={2224--2227},
  year={2005},
  publisher={IEEE}
}

@book {AB09,
    AUTHOR = {Arora, Sanjeev and Barak, Boaz},
     TITLE = {Computational complexity},
      NOTE = {A modern approach},
 PUBLISHER = {Cambridge University Press, Cambridge},
      YEAR = {2009},
     PAGES = {xxiv+579},
      ISBN = {978-0-521-42426-4},
   MRCLASS = {68-01 (03D15 68Q15 68Q17 68Q25 81P68 94A60)},
  MRNUMBER = {2500087},
MRREVIEWER = {Ulrich\ Tamm},
       DOI = {10.1017/CBO9780511804090},
       URL = {https://doi.org/10.1017/CBO9780511804090},
}

@INPROCEEDINGS{RW04,
  author={Renner, R. and Wolf, S.},
  booktitle={International Symposium onInformation Theory, 2004. ISIT 2004. Proceedings.}, 
  title={Smooth {R}enyi entropy and applications}, 
  year={2004},
  pages={233-},
  doi={10.1109/ISIT.2004.1365269}}

@misc{MKS26,
  title={On Sampling Lower Bounds for Polynomials},
  author={Mohammad Mahdi Khodabandeh and Igor Shinkar},
  note={Available at \url{https://eccc.weizmann.ac.il/report/2026/066/}},
  year={2026}
}

\appendix
\crefalias{section}{appendix}

\section{A Non-Constructive Argument}\label{app:non-constructive}

It is often much simpler to show the mere existence of specific objects than to provide explicit examples of them.
While our paper is primarily concerned with these explicit constructions, we record in this appendix a basic existential argument for hard-to-sample distributions for polynomial and $\ac[\oplus]$ sources.
We begin with the following general claim.

\begin{claim}\label{clm:non_constructive_distance}
    If $\Xcal$ is a class of distributions over $\bin^n$ of size $M$, then there exists a (uniform) distribution $\Db$ with
    \[
        \tvdist{\Xb - \Db} \ge 1 - O\pbra{\frac{\log M}{2^n}}^{1/3}
    \]
    for every distribution $\Xb \in \Xcal$.
\end{claim}
\begin{proof}
    Let $\Sb$ be a randomly chosen subset of $\bin^n$ of some size $s$ to be determined, and let $\Db_\Sb$ be the uniform distribution over $\Sb$.
    We will show there exists some choice of $\Sb$ such that $\Db_\Sb$ is far from every distribution in $\Xcal$.

    Consider an arbitrary distribution $\Xb \in \Xcal$.
    Define $L = \cbra{y \in \bin^n : \Xb(y) \le 1/s}$ and $H = \bin^n \setminus L$, noting that $|H| \le s$.
    We have
    \begin{equation*}
        1 - \tvdist{\Xb - \Db_\Sb} = \sum_{y\in \Sb}\min\pbra{\Xb(y), \frac{1}{s}} = \sum_{y\in \Sb \cap L} \Xb(y) + \sum_{y\in \Sb \cap H} \frac{1}{s}.
    \end{equation*}
    For clarity, we define $\Sigmab_L = \sum_{y\in \Sb \cap L} \Xb(y)$ and $\Sigmab_H = \sum_{y\in \Sb \cap H} \frac{1}{s}$.
    Taking expectations, we find that
    \begin{equation*}
        \E_\Sb[\Sigmab_L] = \E_\Sb\sbra{\sum_{y\in \Sb \cap L} \Xb(y)} \le \sum_{y\in L} \Xb(y)\cdot \Pr[y\in \Sb] \le \frac{s}{2^n},
    \end{equation*}
    and similarly that
    \begin{equation*}
        \E_\Sb[\Sigmab_H] = \E_\Sb\sbra{\sum_{y\in \Sb \cap H} \frac{1}{s}} \le \sum_{y\in H} \frac{1}{s}\cdot \Pr[y\in \Sb] \le \frac{|H|}{2^n} \le \frac{s}{2^n}.
    \end{equation*}
    If we apply the version of Hoeffding's inequality for sampling without replacement to the random variables $\cbra{\Xb(y) \cdot \indicator_L(y)}_{y\in \bin^n}$, we find that for any real $t > 0$, 
    \[
        \Pr\sbra{\Sigmab_L \ge \frac{s}{2^n} + t} \le \Pr\sbra{\Sigmab_L - \E[\Sigmab_L] \ge t} \le \exp\pbra{-2t^2 s}.
    \]
    A similar bound holds for $\Sigmab_H$, and combining them with a union bound yields
    \begin{align*}
        \Pr_\Sb\sbra{\tvdist{\Xb - \Db_\Sb} \le 1 - \frac{s}{2^{n-1}} - 2t} &= \Pr_\Sb\sbra{\Sigmab_L + \Sigmab_H \ge \frac{s}{2^{n-1}} + 2t} \\
        &\le \Pr\sbra{\Sigmab_L \ge \frac{s}{2^n} + t} + \Pr\sbra{\Sigmab_H \ge \frac{s}{2^n} + t} \\
        &\le 2\exp\pbra{-2t^2 s}.
    \end{align*}
    An additional union bound over all distributions in $\Xcal$ tells us that
    \begin{equation}\label{eq:prob_close_dist}
        \Pr_\Sb\sbra{\exists \Xb \in \Xcal : \tvdist{\Xb - \Db_\Sb} \le 1 - \frac{s}{2^{n-1}} - 2t} \le 2M\exp\pbra{-2t^2 s}.
    \end{equation}
    
    We conclude the proof by setting $t = c_1\sqrt{\log(M)/s}$ and $s = c_2 \cdot 2^{2n/3}\cdot (\log M)^{1/3}$ for some constants $c_1,c_2 > 0$.
    In this case, the probability in \Cref{eq:prob_close_dist} is strictly less than 1, so there exists a specific $S \subseteq \bin^n$ whose uniform distribution has TV distance $1 - O(\log(M) / 2^n)^{1/3}$ from every distribution in $\Xcal$.
\end{proof}

We cannot immediately obtain hard-to-sample distributions over $\bin^n$ for sources of interest from \Cref{clm:non_constructive_distance}, since they may take arbitrarily many random bits as input, and thus their distribution classes have unbounded size.
Fortunately, these classes can typically be approximated by a collection of sources taking much fewer random bits.
For example, the class of distributions generated by low-degree $\Fbb_2$-polynomials can be approximated by versions taking only $O(n)$ input bits, as we have seen in \Cref{lem:input_redux} (from \cite{chattopadhyay2024extractors}).
We restate this result below for the reader's convenience.

\leminputredux*

From here, one can bound the size of the distribution class, and obtain an optimally hard distribution via \Cref{clm:non_constructive_distance}.

\begin{theorem}
    There exists a constant $\delta > 0$ such that the following statement holds.
    Let $\Xcal$ be the class of $\Fbb_2$-polynomial sources of degree $\delta n$ over $\bin^n$.
    There exists a (uniform) distribution $\Db$ such that
    \[
        \tvdist{\Xb - \Db} \ge 1 - 2^{-\Omega(n)}
    \]
    for every distribution $\Xb \in \Xcal$. 
\end{theorem}
\begin{proof}
    Set $\eps = 2^{- cn}$ for some sufficiently large constant $c > 0$, and let $\ell = \ceilbra{n + 3\log(1/\eps)} = O(n)$.
    The number of degree-$d$ $\Fbb_2$-polynomial sources over $\bin^n$ with at most $\ell$ input bits is $2^{\binom{\ell}{\le d}\cdot n} \le 2^{\pbra{\frac{e\ell}{d}}^d \cdot n}$.
    For $d = \delta n$ and our setting of $\ell$, this is at most $2^{2^{\delta'\cdot n}}$ for some constant $\delta' < 1$, so \Cref{clm:non_constructive_distance} promises a uniform distribution $\Db$ over $\bin^n$ such that
    \[
         \tvdist{Q(\Ub^\ell) - \Db} \ge 1 - 2^{-\Omega(n)}
    \]
    for every degree-$d$ polynomial map $Q\colon \bin^\ell \to \bin^n$.

    We now reduce to the general case of arbitrarily many input bits.
    If $r \le \ell$, then we have already proven the desired lower bound, so assume this is not the case.
    By invoking \Cref{lem:input_redux}, we have that for any degree-$d$ polynomial map $P\colon \bin^r \to \bin^n$, there exists a degree-$d$ polynomial map $Q\colon \bin^\ell \to \bin^n$ such that $Q(\Ub^\ell)$ is $\eps$-close to $P(\Ub^r)$.
    (Note that the transformation $x \mapsto Ax + b$ cannot increase the degree.)
    By our choice of $\eps$, we conclude that any such $P(\Ub^r)$ is also $(1 - 2^{-\Omega(n)})$-far from $\Db$.
\end{proof}

Next, we consider $\ac[\oplus]$ sources.

\begin{theorem}
    There exists a constant $\delta > 0$ such that the following statement holds.
    Let $\Xcal$ be the class of $\ac[\oplus]$ sources of size $2^{\delta n}$ over $\bin^n$.
    There exists a (uniform) distribution $\Db$ such that
    \[
        \tvdist{\Xb - \Db} \ge 1 - 2^{-\Omega(n)}
    \]
    for every distribution $\Xb \in \Xcal$.
\end{theorem}

We cannot apply the same argument as in the case of polynomial sources, since the transformation $x \mapsto Ax + b$ may require $2^{\ell}$ many gates at the start, which is unaffordable.
Fortunately, there is an alternative, elementary analysis to bound the number of inputs.

\begin{proof}
    Let $C\colon \bin^r \to \bin^n$ be an $\ac[\oplus]$ circuit with $g$ gates, and let $\delta' > 2\delta$ be a sufficiently small constant.
    We claim that there exists another such circuit $\tilde{C}\colon \bin^\ell \to \bin^n$ on $\ell \le g(\delta' n + 1)$ many input bits and at most $g$ gates such that $C(\Ub^r)$ is $2^{-\Omega(n)}$-close to $\tilde{C}(\Ub^\ell)$.
    We define $\tilde{C}$ from $C$ in two stages.
    First, fix any \textsf{AND} and \textsf{OR} gates in $C$ which take more than $\delta' n$ original input bits to the constant values 0 and 1, respectively.
    By a union bound, the probability that this changes $C$'s output is less than $2^{-\delta' n}\cdot g \le 2^{-\Omega(n)}$.
    
    Next, consider the subspace $V$ spanned by each of the original input bits feeding into any remaining \textsf{AND} and \textsf{OR} gates and by each of the sums of the original input bits feeding into an \textsf{XOR} gate.
    We claim that we only need as many input bits for $\tilde{C}$ as the dimension of $V$. 
    For this, it suffices to reason about the \textsf{XOR} gates, since the original inputs to the \textsf{AND} and \textsf{OR} gates do not change.
    If any bits feeding into such a gate are independent from the set of input bits we have already constructed, then make the sum of the bits feeding into it a new input bit; otherwise, we can use an appropriate sum of existing inputs.
    
    Setting $\ell = \dim(V) \le g(\delta' n + 1)$ completes the claim that $C(\Ub^r)$ can be approximated by some $\tilde{C}(\Ub^\ell)$, so it remains to bound the number of these latter sources.
    By standard counting arguments (see, e.g., \cite[Section 6.5]{AB09}), the number of size-$g$ $\ac[\oplus]$ sources over $\bin^n$ with at most $\ell$ input bits is at most $2^{O(g(\ell+g) + n\log(\ell + g))}$.
    For $g \le 2^{\delta n}$ and our bound on $\ell$, this is at most $2^{2^{\delta'' \cdot n}}$ for some constant $\delta'' < 1$, so \Cref{clm:non_constructive_distance} promises a uniform distribution $\Db$ over $\bin^n$ such that
    \[
        \tvdist{\tilde{C}(\Ub^\ell) - \Db} \ge 1 - 2^{-\Omega(n)}
    \]
    for every size-$g$ $\ac[\oplus]$ circuit $\tilde{C}\colon \bin^\ell \to \bin^n$.
    Recalling that such sources approximate any size-$g$ $\ac[\oplus]$ circuit to error $2^{-\Omega(n)}$ concludes the proof.
\end{proof}

\end{document}